\documentclass[journal,onecolumn,12pt,twoside]{IEEEtran}
\usepackage{graphicx,multicol}
\usepackage{color,soul}
\usepackage{mathtools}
\usepackage[utf8]{inputenc}
\usepackage{amsmath, amssymb}
\usepackage{mdframed,calc}
\usepackage{amsthm} 
\usepackage{mathrsfs}
\usepackage{cite}
\usepackage[boxed, linesnumbered, ruled]{algorithm2e}
\usepackage{array}
\usepackage[font=small]{caption}
\usepackage{subcaption}
\usepackage{cases}
\usepackage{multirow}
\usepackage[draft]{hyperref}
\usepackage[normalem]{ulem}
\usepackage{empheq}
 \usepackage{tikz,lipsum,lmodern}
\usepackage{diagbox}
\usepackage[most]{tcolorbox}
\SetKwInput{KwInput}{Input}               % Set the Input
\SetKwInput{KwOutput}{Output} 
\SetKwComment{Comment}{/* }{ */} 
\SetAlgoCaptionSeparator{.}          % set the Output
\ifCLASSINFOpdf

\else

\fi

\newtheorem{theorem}{Theorem}%[section]
\newtheorem{definition}{Definition}

\newtheorem{corollary}{Corollary}
\newtheorem{proposition}{Proposition}
\newtheorem{lemma}{Lemma}
\newtheorem{rem}{Remark}

\newcommand{\blue}{\textcolor{black}} %% short forms
\newcommand{\nn}{\nonumber}  %% nonumber
\newcommand{\bieee}{\begin{eqnarray}{rCl}}
\newcommand{\eieee}{\end{eqnarray}}

\usepackage[nodisplayskipstretch]{setspace}
\begin{document}
\title{Low-Latency Communication using Delay-Aware Relays Against Reactive Adversaries}
\author{Vivek~Chaudhary and
        Harshan~Jagadeesh% <-this % stops a space
\thanks{V. Chaudhary and H. Jagadeesh are with the Department
of Electrical Engineering, Indian Institute of Technology, Delhi, 110016, India. E-mail: (chaudhary03vivek@gmail.com, jharshan@ee.iitd.ac.in)}% <-this % stops a space
}
\maketitle
\begin{abstract}
This work addresses a reactive jamming attack on the low-latency messages of a victim, wherein the jammer deploys countermeasure detection mechanisms to change its strategy. We highlight that the existing schemes against reactive jammers use relays with instantaneous full-duplex (FD) radios to evade the attack. However, due to the limitation of the radio architecture of the FD helper, instantaneous forwarding may not be possible in practice, thereby leading to increased decoding complexity at the destination and a high detection probability at the adversary. Pointing at this drawback, we propose a delay-aware cooperative framework wherein the victim seeks assistance from a delay-aware FD helper to forward its messages to the destination within the latency constraints. In particular, we first model the processing delay at the helper based on its hardware architecture, and then propose two low-complexity mitigation schemes, wherein the victim and the helper share their uplink frequencies using appropriate energy-splitting factors. For both the schemes, we solve the optimization problems of computing the near-optimal energy-splitting factors that minimize the joint error rates at the destination. Finally, through analytical and simulation results, we show that the proposed schemes facilitate the victim in evading the jamming attack whilst deceiving the reactive adversary.\looseness = -1
\end{abstract}
%
%% Note that keywords are not normally used for peerreview papers.

\IEEEpeerreviewmaketitle

\section{Introduction}

\textcolor{black}{The next generation of wireless networks finds its use-cases in critical infrastructures such as vehicular networks involving autonomous vehicles \cite{wireless_security}. Since these applications carry vital information that needs to reach the destination within a deadline, they are ideal targets for an adversary. Among the various attack models that target use-cases with deadline constraints, jamming attacks, due to their ease of execution using off-the-shelf radio devices, have been popular means of executing Denial of Service (DoS) attacks thereby forcing deadline violation on the packets}. Moreover, due to recent technological advancements in radio architecture, the adversary has become more potent than its traditional counterpart. In particular, it has been shown that in addition to jamming, an adversary may monitor the network for possible countermeasures and change its attacking strategy based on the action taken by the victim. Such a class of adversaries are referred to as reactive adversaries \cite{reactive_jammer1, reactive_jammer2}. \blue{Using the recent developments in full-duplex (FD) radios \cite{FD1,FD2,FD3,FD4,FD5,FD6,FD7} and cognitive radios \cite{FDCR1,FDCR2,FDCR3,FDCR4}, the authors in  \cite{my_PIMRC, TCCN-VH, my_GCOM, my_comsnets} presented one such adversarial model, where a reactive adversary uses FD cognitive radios to execute jamming attacks.} Here, besides jamming, the adversary measures the average energy level of the jammed frequency to prevent the victim node from using state-of-art countermeasures such as, frequency hopping (FH). Also, \cite{my_TCOM} presented another class of adversarial model, wherein the reactive jamming adversary also measures the correlation between the symbols of the jammed frequency and other frequencies, thus, preventing the victim from using repetition coding across frequencies. As a first step towards mitigating the energy monitoring reactive adversary, \cite{my_PIMRC, TCCN-VH, my_GCOM,my_comsnets} presented a decode-and-forward based cooperative countermeasure, wherein the victim uses a fraction of its energy to communicate its low-latency messages to the destination with the help of an adjacent FD helper. The helper multiplexes the victim's symbols along with its symbols so as to facilitate the destination to jointly decode the symbols of both the nodes. Furthermore, the countermeasure is such that the victim and the helper pour their residual energies to ensure not getting detected by the adversary. We point out that the foundational assumption for the analysis in \cite{my_PIMRC, TCCN-VH, my_GCOM, my_comsnets} is that the helper instantaneously decodes, and multiplexes the victim's decoded symbols to the destination thereby ensuring that the messages of the victim do not violate the deadline constraint. We note that while instantaneous forwarding by the helper facilitates low-latency communication of the victim's messages, it may not be realizable by all radio architectures due to the processing delay. This opens up questions on how to design countermeasures with practical FD radios and still ensure that the victim's messages reach the destination within the deadline.

\subsection{Motivation} 

In the countermeasure contributed by \cite{my_PIMRC, TCCN-VH, my_GCOM}, it is assumed that the processing delay at the helper in forwarding the victim's symbols is negligible. Thus, the victim's symbols and the multiplexed symbols from the helper reach the destination during the same \blue{symbol interval}. However, when using practical FD radios, the forwarding process may not be instantaneous, and the processing delay at the helper can be of the order of several symbol durations. In such cases, the symbols on the two links reach the destination  during different \blue{symbol intervals}, thereby yieliding a signal model different from that of \cite{my_PIMRC, TCCN-VH, my_GCOM}. Moreover, if the helper decides to use multiple receive-antennas to improve the diversity order, the processing delay further increases due to additional delay contributed by the self-interfernce cancellation (SIC) blocks of the FD radios \cite{FD_MIMO}. Thus, using practical FD radios to combat reactive adversaries has the following three consequences: (i) The existing analysis of \cite{my_PIMRC, TCCN-VH, my_GCOM} does not hold as the symbols on the Victim-to-Destination and the Helper-to-Destination links are observed at the destination in an asynchronous fashion, (ii) Due to the processing delay, a few multiplexed symbols reach the destination after the deadline, thereby violating the low-latency constraint, and (iii) The symbols on the victim's and the helper's frequencies are uncoordinated in energy, resulting in fluctuations in the average energy level of both the frequencies, thereby increasing the probability of detection by the energy detector. \emph{Thus, these limitations of the existing countermeasures motivate us to design new countermeasures that consider the helper's practical limitations in facilitating reliable and low-latency communication of the victim's symbols.}

\subsection{Contributions}
\begin{enumerate}
\item To facilitate reliable communication between the victim and the destination, we propose a framework wherein the victim seeks assistance from an FD helper to multiplex-and-forward its symbols to the destination in an asynchronous manner, such that the victim and the helper share the helper's uplink frequency using an energy-splitting factor $\alpha\in(0,1)$. We first model the processing delay at the helper using the parameter $\Theta$ which is a function of the number of receive-antennas at the helper and then propose a strategy so that the victim and the helper cooperatively use their bands to reliably communicate and still not get detected by the adversary. Since this framework incorporates the delay parameter at the helper, we refer to this framework as the Delay-Aware Semi-Coherent Multiplex-and-Forward (DASC-MF) mitigation scheme. With On-Off Keying (OOK) at the victim and $M-$PSK at the helper, we highlight that due to the processing delay, the symbols received across several \blue{symbol intervals} at the destination are correlated. Thus, the decoding complexity of the optimal decoder is $\mathcal{O}(4M^{2})$, which makes its implementation challenging. We also show that the symbols received at the destination from the victim and the helper are still uncoordinated in energy, thereby, making the proposed countermeasure susceptible to detection by the adversary. To circumvent these challenges, we propose $3\phi$ DASC-MF scheme, which falls under the framework of DASC-MF as a special case. (See Sec.~\ref{sec:DASC-MF})

\item In the $3\phi$ DASC-MF scheme, we divide the frame structure into three parts, parametrized by the processing delay, $\Theta$, such that, $\Theta\leq \frac{L}{2}$, where $L$ denotes the number of symbols transmitted in a frame by the victim. The novel idea of this strategy is to use two energy-splitting factors, $\alpha\in(0,1)$ and $\beta\in(0,1)$ at different portions of the frame. Through an appropriate choice of $\alpha$, we show that the correlation across the symbols can be minimized thereby ensuring that symbols at different \blue{symbol intervals} can be independently decoded and also show improved energy coordination when compared to the vanilla DASC-MF scheme. For this strategy, we provide strong analytical results on the error performance, and based on these results, we provide a near-optimal solution on $\alpha$ and $\beta$ to the optimization problem of minimizing the error-rates. We also show that $3\phi$ DASC-MF is less complex than DASC-MF, and a majority of the symbols transmitted during $3\phi$ DASC-MF scheme are coordinated in energy. (See Sec.~\ref{sec:Three_Phase})

\item When $\Theta>\frac{L}{2}$, we propose a new countermeasure, referred to as the semi-coherent multiple access channel (SC-MAC) scheme. As a salient feature of this scheme, the helper does not decode the victim's symbols, instead, the victim and the helper transmit their symbols synchronously to the destination on the helper's frequency using an energy-splitting factor, $\varepsilon\in(0,1)$, thereby, eliminating the need of an FD radio at the helper. For SC-MAC, we first derive a closed-form expression on the error-rates and then, formulate an optimization problem of finding near-optimal values on $\varepsilon$ that minimizes the error-rates at the destination. (See Sec.~\ref{sec:sc-mac})
\item \textcolor{black}{We also present extensive simulation results to show that using both the $3\phi$ DASC-MF scheme and the SC-MAC scheme, the victim is able to reliably communicate with the destination while adhering to the deadline constraints. (See Sec.~\ref{sec:sim})} 
\item Finally, through various analytical and simulation results, we show that our schemes are covert when the adversary measures energy on the victim's and the helper's frequencies. (See Sec.~\ref{sec:Covertness})
\end{enumerate}

%%%%%%%% New Related Works
\subsection{\blue{Related Work}}

\blue{Due to the recent technological advancements in the FD radio architectures \cite{FD4,FD6,FD5,FD3,FD2,FD1,FD7}, FD radios have been studied from the viewpoints of mitigating  adversaries \cite{FD8, Aid_FD_1,Aid_FD_2,Aid_FD_3} as well as aiding adversaries \cite{Foe_FD_1, conv_attack, Foe_FD_3}. However, \cite{my_PIMRC, TCCN-VH, my_comsnets, my_GCOM, my_TCOM} have studied FD radios from both these viewpoints wherein an FD adversary is used to jam the network and an FD helper node is used to mitigate the FD jammer. When using an FD radio at the adversary, \cite{my_PIMRC, TCCN-VH, my_comsnets} presented \emph{jam and measure} adversaries, that jams a frequency band and subsequently monitors it using FD cognitive radios \cite{FDCR1, FDCR2, FDCR3, FDCR4} to detect countermeasures based on FH. Moreover, \cite{my_PIMRC, TCCN-VH} used countermeasure detectors based on energy measurement, while \cite{my_comsnets} used countermeasure detectors based on energy and correlation. In contrast, when using an FD radio at the helper node, \cite{my_PIMRC, TCCN-VH, my_comsnets} also proposed fast-forward FD relay based countermeasures to mitigate the \emph{jam and measure} adversaries. Here, the authors leveraged on \cite{FD8}, which proposed a fast-forward FD relay that constructively forwards signals such that the network throughput and coverage is significantly enhanced. Along the similar lines of \cite{my_PIMRC, TCCN-VH, my_comsnets}, authors in \cite{my_GCOM, my_TCOM} proposed fast-forward relaying based solutions to mitigate the \emph{jam and measure} adversary, however in fast-fading channel conditions. We highlight that, when mitigating \emph{jam and measure} adversaries, \cite{my_PIMRC, TCCN-VH, my_comsnets, my_GCOM, my_TCOM} assumed an optimistic scenario where the FD helper node instantaneously fast-forwards the victim's information symbol to the destination. However, in practice since the processing delay for self-interference cancellation is directly proportional to the number of transmit/receive (or both) antennas at the helper node \cite{FD_MIMO}, the fast-forwarding process is not instantaneous. Thus, to bridge this gap, this work considers the practical limitations of delay-aware FD radios in facilitating reliable and low-latency communication of the victim's symbols. As depicted in Fig.~\ref{fig:venn_diagram}, the main novelty of this work is the use of delay-aware FD radios at the helper node, which has not been addressed in the literature hitherto.}

\begin{figure}
\vspace{-0.5cm}
\centering
\includegraphics[scale=0.35]{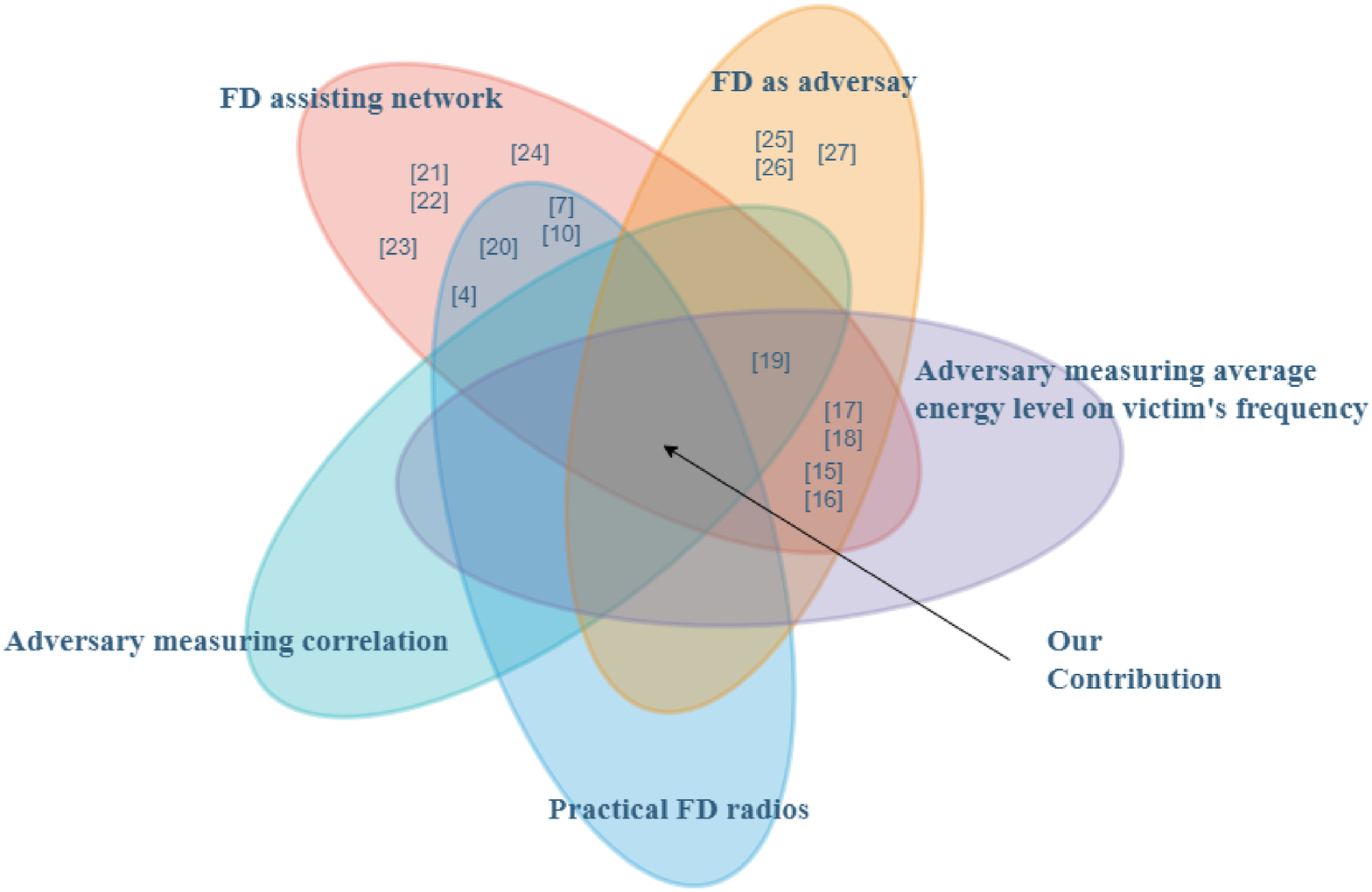}
\vspace{-0.25cm}
\caption{\label{fig:venn_diagram} Novelty of our work w.r.t. existing contributions.}
\vspace{-0.5cm}
\end{figure}

\section{System Model}
\label{sec:system}

%\begin{figure}
%\centering
%\includegraphics[width = 0.6\textwidth, height = 0.2\textwidth]{system_model}
%\caption{\label{fig:system_model} (a) A network model depicting legitimate nodes, Alice and Charlie, and the reactive adversary, Dave. (b) System model for DASC-MF scheme, where Charlie takes $\Theta$ symbols to multiplex-and-forward Alice's symbols to Bob.}
%\end{figure}

\begin{figure}[b]
\centering
\includegraphics[scale = 0.35]{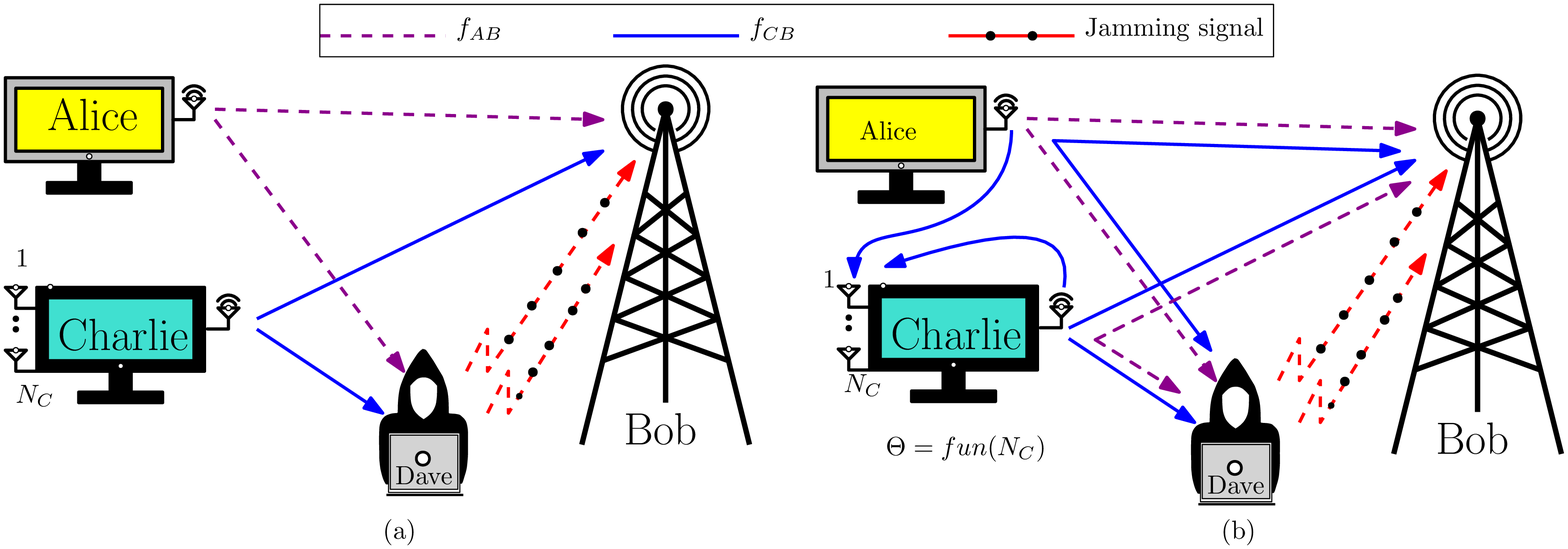}
\caption{\label{fig:system_model} (a) A network model depicting legitimate nodes, Alice and Charlie, and the reactive adversary, Dave. (b) System model for DASC-MF scheme, where Charlie takes $\Theta$ symbols to multiplex-and-forward Alice's symbols to Bob.}
\end{figure}

We consider a network model, where all the uplink frequencies to the destination are occupied by the legitimate nodes of the network. As a result, there are no vacant frequencies in the network. Let Alice and Charlie be two legitimate nodes of the network who communicate with the destination, Bob using orthogonal frequencies. Alice is a single-antenna transmitter which uses a frequency band centred around $f_{AB}$ to communicate her low data-rate symbols with strict low-latency constraints (e.g., PUCCH in 5G \cite{standard}) to Bob. In contrast, Charlie uses a frequency band centred around $f_{CB}$ to communicate his symbols to Bob. A salient feature of Charlie is that he is equipped with a full-duplex (FD) radio with $N_{C}$ receive-antennas and a single transmit-antenna, thus, enabling him to implement FD communication on $f_{CB}$. Further, Charlie transmits symbols with arbitrary data-rate and relaxed latency constraints. An instantiation of the network is as shown in Fig.~\ref{fig:system_model}~(a).

In this network, we also assume the presence of a jammer, Dave. Unlike traditional jammers, Dave is an FD reactive jammer who intends to jam Alice's uplink symbols on $f_{AB}$ \cite{DOSattack} and  monitors all the frequencies (including $f_{AB}$) for possible countermeasures by the legitimate nodes. In the context of this work, Dave uses an energy detector (ED) to measure the average energy level of all the frequencies in the network before and after jamming. Thus, if ED measures a significant fluctuation in the average energy level on any band, it raises the flag. This in turn forbids Alice from using traditional mitigation techniques, such as FH. Subsequently, a raised flag by ED compels Dave to jam other frequencies, thereby degrading the network's performance. Thus, in the next section, we propose a framework wherein Charlie assists Alice to reliably communicate her messages to Bob without getting detected by ED on any frequency band.
 
\section{Delay-Aware Semi-Coherent Multiplex-and-Forward Relaying Scheme}
\label{sec:DASC-MF}

\blue{As shown in Fig.~\ref{fig:dummy_scheme}, let Alice have a frame of $L$ symbols to communicate with Bob within $T$ seconds from the generation of the first symbol, where $T=L/W$, such that $W$ is the bandwidth of communication.} Since $f_{AB}$ is jammed by Dave, Alice seeks help from Charlie. As part of the protocol, Alice broadcasts her symbols on $f_{CB}$. Then, Charlie, uses his FD radio to forward Alice's symbols to Bob on $f_{CB}$. Since Charlie is a legitimate node in the network, he also has symbols to communicate with Bob. Therefore, he decodes Alice's symbols, multiplexes her symbols to his symbols, and forwards them to Bob. However, the time taken by Charlie for this process depends on his receiver architecture. In particular, this \emph{delay} is directly proportional to the time taken by him to cancel his self-interference (SI), which in turn is directly proportional to the number of receiving antennas, $N_{C}$. Therefore, we assume that Charlie requires a time duration equivalent to that of $\Theta$ symbols to decode, multiplex, and forward Alice's symbols, where $\Theta$ is governed by the SIC architecture of Charlie. 

The top two frames in Fig.~\ref{fig:dummy_scheme} show the symbols transmitted simultaneously by Alice and Charlie on $f_{CB}$. The received symbols at Bob, denoted by $r_{B,n}$ are captured by the bottom frame, where $n$ denotes the \blue{symbol interval} index for communication. Further, $x_{n}$ denotes Alice's symbol, and $y_{n}$, for $1\leq n\leq \Theta$ and $t_{n}$, for $\Theta+1\leq n\leq L+\Theta$ denote the unmultiplexed and multiplexed symbols transmitted by Charlie, respectively. Due to delay of $\Theta$ symbols, $r_{B,n}$ is a function of $x_{n}$ and $y_{n}$, for $1\leq n\leq \Theta$. In addition, $r_{B,n}$ is a function of $x_{n-\Theta}$, $x_{n}$, and $y_{n}$, for $\Theta+1\leq n\leq L$. Note that, since the multiplexed symbols, $t_{n}$, for $L+1\leq n\leq L+\Theta$ are received after the deadline of $T$ seconds, Bob cannot use these symbols for decoding Alice's symbols due to latency constraints. Therefore, Bob only uses the first $L$ symbols received on $f_{CB}$ after implementing the proposed countermeasure to jointly decode Alice's and Charlie's symbols.\footnote{During implementation, Charlie may decide not to multiplex after the $L$-th \blue{symbol interval} since Alice's symbols are no longer used for decoding from $r_{B, n}$ for $n > L$}

\begin{figure}
\vspace{-0.5cm}
\centering
\includegraphics[scale=0.28]{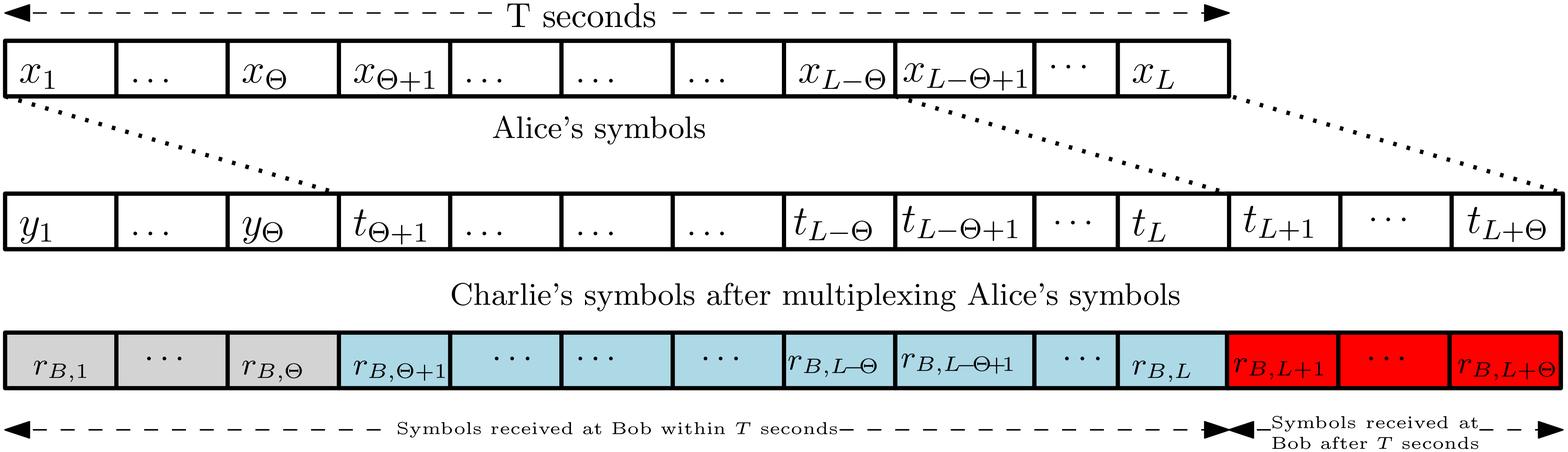}
\caption{\label{fig:dummy_scheme} Illustration of symbol transmission in DASC-MF scheme.}
\vspace{-0.5cm}
\end{figure}

Although the above scheme ensures that Alice's symbols reach Bob within the deadline, Dave observes a significant drop in the energy level on $f_{AB}$, provided Alice uses her entire energy to communicate on $f_{CB}$. Therefore, to avoid getting detected by Dave, Alice and Charlie use $1-\alpha$ and $\alpha$ fractions of their energies, respectively, to communicate their messages on $f_{CB}$, where $\alpha\in(0,1)$ is the design parameter under consideration. Subsequently, Alice and Charlie also use $\alpha$ and $1-\alpha$ fractions of their energies, respectively, to cooperatively transmit dummy OOK symbols on $f_{AB}$. The dummy OOK symbols are sampled from a pre-shared Gold sequence, thus, ensuring that the average energy level on $f_{AB}$ is the same as it was before implementing the countermeasure. We highlight that the use of dummy OOK symbols also ensures that the symbols on $f_{AB}$ and $f_{CB}$ are uncorrelated. Fig.~\ref{fig:system_model}~(b) illustrates the transmission strategy by Alice and Charlie on both $f_{AB}$ and $f_{CB}$.

Due to low-latency constraints, Alice avoids the use of pilots for assisting channel estimation on $f_{CB}$, as a result, the communication on the links Alice-to-Charlie and Alice-to-Bob are inherently non-coherent. Further, since Alice has low data-rate symbols, she uses OOK signalling scheme. In contrast, since $f_{CB}$ is not under attack, Charlie uses a coherent signalling scheme. In particular, Charlie uses $M-$ary PSK to communicate with Bob. As a result, the proposed scheme is a combination of coherent and non-coherent communication under a multiple access channel setup. Further, since this scheme takes into account the delay in processing Alice's symbols at Charlie, we refer to the proposed scheme as Delay-Aware Semi-Coherent Multiplex-and-Forward (DASC-MF) scheme.  

In the next section, we discuss the signal model of DASC-MF scheme on $f_{CB}$. A detailed description of communication on $f_{AB}$ and the analysis on the energy detectors will be discussed in Sec.~\ref{sec:Covertness}.

\subsection{Signal Model}

In the DASC-MF scheme, Alice transmits $x_{n}\in\{0,1\}$ scaled by $\sqrt{1-\alpha}$ throughout the entire frame. Since Charlie is equipped with an FD radio with $N_{C}$ receive-antennas, the $N_{C}\times 1$ received vector at Charlie is given by
\bieee
\mathbf{r}_{C,n} = \sqrt{1-\alpha}\mathbf{h}_{AC,n}x_{n} + \mathbf{h}_{CC,n} + \mathbf{w}_{C,n}, \hspace{0.5 cm} 1\leq n\leq L,\label{eq:rc}
\eieee
\noindent where $\mathbf{h}_{AC,n}\sim\mathcal{CN}\left(\mathbf{0}_{N_{C}}, \sigma_{AC}^{2}\mathbf{I}_{N_{C}}\right)$, $\mathbf{h}_{CC,n}\sim\mathcal{CN}\left(\mathbf{0}_{N_{C}}, \lambda\frac{1+\alpha}{2}\mathbf{I}_{N_{C}}\right)$, and $\mathbf{w}_{C,n}\sim\mathcal{CN}\left(\mathbf{0}_{N_{C}}, N_{o}\mathbf{I}_{N_{C}}\right)$ are the $N_{C}\times 1$ vectors of fading channel coefficients of Alice-to-Charlie's link, residual SI of Charlie's FD radio, and the Additive White Gaussian Noise (AWGN) at Charlie, respectively, such that $\mathbf{0}_{N_{C}}$ is the $N_{C}\times 1$ vector of zeros and $\mathbf{I}_{N_{C}}$ is the $N_{C}\times N_{C}$ Identity matrix. Further, $\lambda\in[0,1]$ denotes the level of residual SI after the active and passive cancellations at Charlie. Finally, $N_{o}=$ SNR\textsuperscript{-1} is the variance of each element of the AWGN vector. 

Charlie first uses non-coherent energy detection to decode $x_{n}$ as $\hat{x}_{n}$ using the received vector, $\mathbf{r}_{C,n}$ and then multiplexes the decoded symbol, $\hat{x}_{n}$ to his symbol. However, due to the use of multiple receive-antennas, Charlie introduces a delay  equivalent to $\Theta$ symbols to decode and multiplex $\hat{x}_{n}$ to his current symbol after SIC. Therefore, a symbol received at Charlie from Alice during the $n^{th}$ \blue{symbol interval} is multiplexed and forwarded to Bob during the $(n + \Theta)^{th}$ \blue{symbol interval}. As a result, for all $1\leq n\leq \Theta$, if $y_{n}\in\mathcal{S}_{C}$ denotes the original PSK symbol of Charlie to be transmitted at $n^{th}$ \blue{symbol interval}, such that $\mathcal{S}_{C} = \left\{\left.e^{\iota\frac{\pi}{M}\left(2m+1\right)}\right\vert m = 0,\ldots, M-1\right\}$, Charlie transmits $\sqrt{\alpha}y_{n}$. Subsequently, for all $n\geq \Theta +1$, Charlie transmits the multiplexed symbol, $t_{n}$, such that 
\begin{subnumcases}{t_{n} =}
 y_{n}, & if $\hat{x}_{n-\Theta} = 0$,\label{eq:tn1}
\\
\sqrt{\alpha}e^{\iota\frac{\pi}{M}}y_{n}, & if $\hat{x}_{n-\Theta} = 1$.\label{eq:tn2}
\end{subnumcases}

With this transmission strategy, Bob observes a multiple access channel on $f_{CB}$ from Alice and Charlie. We note that Alice's symbol transmitted during the $n^{th}$ \blue{symbol interval} is observed at Bob during $n^{th}$ time symbol on Alice-to-Bob link on $r_{B,n}$ and is again observed as multiplexed symbol from Charlie at $(n+\Theta)^{th}$ time symbol on Charlie-to-Bob link on $r_{B,n+\Theta}$, thus, introducing correlation between $r_{B,n}$ and $r_{B,n+\Theta}$. Overall, the received symbol at Bob during $n^{th}$ \blue{symbol interval} is given as
\begin{subnumcases}{r_{B,n} =}
 \sqrt{1-\alpha}h_{AB,n}x_{n} + \sqrt{\alpha}h_{CB,n}y_{n} + w_{B,n}, & if $1\leq n\leq \Theta$,\label{eq:rB1}
\\
\sqrt{1-\alpha}h_{AB,n}x_{n} + h_{CB,n}t_{n} + w_{B,n}, & if $\Theta+1\leq n\leq L$,\label{eq:rB2}
\end{subnumcases}
\noindent where $h_{AB,n}\sim\mathcal{CN}(0,\sigma_{AB}^{2})$ and $h_{CB,n}\sim\mathcal{CN}(0,\sigma_{CB}^{2})$ are the channel coefficients of Alice-to-Bob and Charlie-to-Bob links, respectively. Further, $w_{B,n}\sim\mathcal{CN}(0,N_{o})$ is the AWGN at Bob during the $n^{th}$ \blue{symbol interval}. For the decoding process, we assume that Bob has the knowledge of $h_{CB,n}$. Further, since Charlie is in the vicinity of Alice, we assume $\sigma_{AC}^{2}> \sigma_{AB}^{2}$, thus achieving a higher SNR on Alice-to-Charlie link as compared to Alice-to-Bob link. We also assume that all channels and noise realizations are statistically independent. Finally, for error analysis in the rest of the paper, we use $\sigma_{AB}^{2}=\sigma_{CB}^{2} = 1$.\looseness = -1
  
\subsection{Error Analysis at Bob}

Since Charlie decodes and multiplexes Alice's symbols to Bob, in this section, we first characterise the error introduced by Charlie in decoding Alice's symbols. We then compute the joint error-rates at Bob in decoding Alice's and Charlie's symbols.

Based on \eqref{eq:rc}, the maximum likelihood (ML) decoder for detection of Alice's symbols at Charlie is
\bieee
\hat{x}_{n} &=& \arg\underset{i\in\{0,1\}}{\max\ }g\left(\mathbf{r}_{C,n}\vert x_{n}=i\right) =  \arg\underset{i\in\{0,1\}}{\min\ } N_{C}\ln(\pi\Omega_{i}) + \frac{\mathbf{r}_{C,n}^{H}\mathbf{r}_{C,n}}{\Omega_{i}},\label{eq:ML_Charlie}
\eieee
\noindent where $g\left(\mathbf{r}_{C,n}\vert x_{n}\right)$ is the probability density function (PDF) of $\mathbf{r}_{C,n}$ conditioned on $x_{n}$ and $\Omega_{i} = \sigma_{AC}^{2}(1-\alpha)i+\lambda\frac{(1+\alpha)}{2} + N_{o}$. Based on \eqref{eq:ML_Charlie}, the threshold for energy detection is given by $\tau = N_{C}\frac{\Omega_{0}\Omega_{1}}{\Omega_{0}-\Omega_{1}}\ln\left(\frac{\Omega_{0}}{\Omega_{1}}\right)$. Finally, using $\tau$, it is straightforward to prove the next theorem that presents the probability of error at Charlie in decoding Alice's symbols.
\begin{theorem}
\label{th:P01P10}
If $\Phi_{i\overline{i}}$ denotes the probability of decoding symbol $i$ as $\overline{i}$, for $i,\overline{i}=\{0,1\}$, then $\Phi_{01} = \frac{\Gamma\left(N_{C}, \frac{\tau}{\Omega_{0}}\right)}{\Gamma(N_{C})}$ and $\Phi_{10} = \frac{\gamma\left(N_{C}, \frac{\tau}{\Omega_{1}}\right)}{\Gamma(N_{C})}$, where $\Gamma(\cdot,\cdot)$ and $\gamma(\cdot,\cdot)$ denote the lower and upper incomplete gamma functions, respectively, and $\Gamma(\cdot)$ denotes the complete gamma function.
\end{theorem}

From Theorem~\ref{th:P01P10}, we immediately observe the following two remarks.
\begin{rem}
\label{rem:P01P10alpha}
The terms $\Phi_{01}$ and $\Phi_{10}$ are increasing functions of $\alpha$ for a given $N_{C}$, $SNR$, and $\lambda$.
\end{rem}
\begin{rem}
\label{rem:P01P10Nc}
The terms $\Phi_{01}$ and $\Phi_{10}$ are decreasing functions of $N_{C}$, for a given $\alpha$, $SNR$, and $\lambda$.
\end{rem}

Based on \eqref{eq:rB1} and \eqref{eq:rB2}, the joint maximum a posteriori (MAP) decoder for DASC-MF scheme is
\begin{small}
\bieee
\hat{i},\hat{j},\hat{k}, \hat{l} =\arg\underset{\substack{i,j,k,l}}{\max\ }g_{B}\left(r_{B,n}, r_{B,n-\Theta}\vert x_{n} = i, x_{n-\Theta}=j, y_{n} = e^{\iota\frac{\pi}{M}\left(2l+1\right)}, y_{n-\Theta} = e^{\iota\frac{\pi}{M}\left(2k+1\right)}, h_{CB,n}, h_{CB,n-\Theta}\right),\label{eq:MAP}
\eieee
\end{small}
\noindent where $i,j\in\{0,1\}$, $k,l\in\{0,M-1\}$ and $g_{B}(.)$ is the joint PDF of $r_{B,n}$ and $r_{B,n-\Theta}$ conditioned on $x_{n}$ and $y_{n}$ for $1\leq n\leq \Theta$ and $x_{n}$, $x_{n-\Theta}$, $y_{n-\Theta}$, and $y_{n}$, for $\Theta+1\leq n \leq L$.

Towards decoding Alice's and Charlie's symbols, the implementation of the decoder in \eqref{eq:MAP} is complex due to the correlation between the symbols received $\Theta$ symbols apart. In particular, the complexity of the proposed decoder is $\mathcal{O}(4M^{2})$ as Bob has to jointly decode two OOK and two PSK symbols. Further, we note that, at any instant the symbols transmitted by Alice and Charlie are uncoordinated in their energies. For instance, when Alice transmits bit-0, Alice and Charlie contribute zero and $\alpha$ energies, respectively. Thus, the resultant sum energy on $f_{CB}$ is $\alpha$. As a result, Dave's ED that is monitoring $f_{CB}$ may observe a dip in the average energy level on $f_{CB}$. Therefore, to circumvent the above problems, in the next section, we propose a variation of DASC-MF scheme, such that the new scheme is amenable to lower-decoding complexity at Bob. Further, the new scheme also ensures that despite uncoordinated transmission from Alice and Charlie, the duration for which Dave's ED observes a dip in the average energy level on $f_{CB}$ is small.

 \section{$3\phi$ Delay Tolerant Semi-Coherent Multiplex-and-Forward Relaying Scheme}
\label{sec:Three_Phase}

From the discussions in the previous section, we note that Alice's information symbol, $x_{n}$, $1\leq n\leq L-\Theta$ is observed twice at Bob during the interval of $T$ seconds: once during the \blue{symbol interval} $1\leq n\leq L-\Theta$ on the Alice-to-Bob link of the MAC and again after $\Theta$ symbols as $t_{n}$, for $\Theta+1\leq n\leq L$ on the Charlie-to-Bob link of the MAC. Therefore, if Bob discards Alice's symbols on Alice-to-Bob link of the MAC for $1\leq n\leq L-\Theta$ and treat these symbols as interference, he can still recover these symbols using the multiplexed symbols $t_{n}$, $\Theta+1\leq n\leq L$. Furthermore, if the interference caused from Alice's symbols on the \blue{symbol intervals} $1\leq n\leq L-\Theta$ are somehow suppressed, then the correlation between $r_{B,n}$ and $r_{B,n+\Theta}$ can be minimized, thereby ensuring that Charlie's symbols on $r_{B,n}$ are decoded independent of Charlie's multiplexed symbols on the other \blue{symbol intervals}. Subsequently this would facilitate reduced decoding complexity at Bob. To facilitate interference suppression, we propose a method of choosing $\alpha$ for the \blue{symbol intervals} $1\leq n\leq L-\Theta$, such that Charlie would continue to reliably recover Alice's symbols for multiplexing process. We note that, since the Alice-to-Bob link of the MAC is non-coherent, Alice contributes $1-\alpha$ and zero energies on this link for $x_{n}=1$ and $x_{n}=0$, respectively. As a result, the variance of the effective noise at Bob is utmost $N_{o}+1-\alpha$ and $N_{o}$ for $x_{n}=1$ and $x_{n}=0$, respectively. Since $N_{o}+1-\alpha$ is a decreasing function of $\alpha$, if we increase $\alpha$ close to $1$, we can suppress the interference on Alice-to-Bob link of the MAC, when $x_{n}=1$ is sent from Alice. In particular, if $1-\alpha = \Delta N_{o}$, such that $0<\Delta\ll 1$, where $\Delta$ is the design parameter, then, $N_{o}+1-\alpha = N_{o}(1+\Delta)\approx N_{o}$. However, we must note that, when $\alpha$ is close to $1$, Charlie requires a large $N_{C}$ to reliably decode Alice's symbols (Remark~\ref{rem:P01P10alpha}). Therefore, if we indefinitely reduce $\Delta$ to a very small value to increase $\alpha$ close to $1$, $N_{C}$ increases which in turn increases the latency at Charlie. Thus, in our proposed scheme, interference suppression at Bob comes at a cost of large $N_{C}$.

 From the above discussion, the transmission scheme for the \blue{symbol intervals}, $1\leq n\leq L$ at Alice and Charlie can be divided into three phases. During Phase-I, $1\leq n\leq \Theta$, Alice and Charlie transmit their symbols scaled by $1-\alpha$ and $\alpha$ fractions of their energies, respectively. Subsequently, during Phase-II, $\Theta+1\leq n\leq L-\Theta$, Alice continues to transmit her symbols scaled by $1-\alpha$ fraction of her energy, however, Charlie scales the multiplexed symbol by $\alpha$ fraction of his energy as given in \eqref{eq:tn1}~-~\eqref{eq:tn2}. Further, due to processing delay of $\Theta$ symbols at Charlie, the multiplexed symbols corresponding to $x_{n}$, $L-\Theta+1\leq n\leq L$, reach Bob after the deadline i.e., after $T$ seconds. Therefore, these symbols cannot be decoded using the multiplexed symbols and instead must be decoded using the symbols on Alice-to-Bob link of the MAC. Thus, for $L-\Theta+1\leq n\leq L$, Bob must jointly decode three symbols, i.e.,  Alice's current symbol, $x_{n}$, Charlie's current symbol, $y_{n}$, and multiplexed Alice's symbol, $x_{n-\Theta}$. Since these symbols are transmitted via combination of coherent and non-coherent modulation schemes, Bob needs distinguishable energy levels for detection when $x_{n}=0$ and $x_{n}=1$ is sent. As a consequence, for $L-\Theta+1\leq n\leq L$, we cannot use $\alpha = 1-\Delta N_{o}$ and instead use a different energy-splitting factor, $\beta\in(0,1)$. Therefore, we refer to the \blue{symbol intervals} $L-\Theta+1\leq n\leq L$ as Phase-III, wherein, Alice and Charlie transmit their  symbols scaled by $1-\beta$ and $\beta$ fraction of their energies, respectively. Here, Charlie only rotates his PSK symbol by $e^{\iota\frac{\pi}{M}}$ when he decodes symbol $1$ from Alice. It is evident from the discussions that the maximum delay tolerated by the proposed $3\phi$ DASC-MF scheme is $\frac{L}{2}$, i.e., $\Theta \leq \frac{L}{2}$. This is because, for $\Theta>\frac{L}{2}$, only a fraction of Alice's symbols are recoverable using Charlie's multiplexed symbols and a majority of the multiplexed symbols are received after the deadline, thus, violating the deadline constraint.
 
 Overall, the symbols received at Bob during each phase are tabulated in Table~\ref{tab:three_phase}. Further, in Fig.~\ref{fig:frame_model}, the top two frames depict the symbols transmitted by Alice and Charlie when using the $3\phi$ DASC-MF scheme. The bottom frame depicts the corresponding symbols received at Bob during each phase. Furthermore, assuming Charlie transmits symbols using $4-$PSK signalling, the constellation diagrams jointly contributed by Alice and Charlie during each phase at Bob are shown in Fig.~\ref{fig:cons}.
 
\begin{figure*}[!htb]
\captionsetup{width=0.48\textwidth}
\begin{center}
    \begin{minipage}{0.48\textwidth}
    \centering
        \includegraphics[scale=0.25]{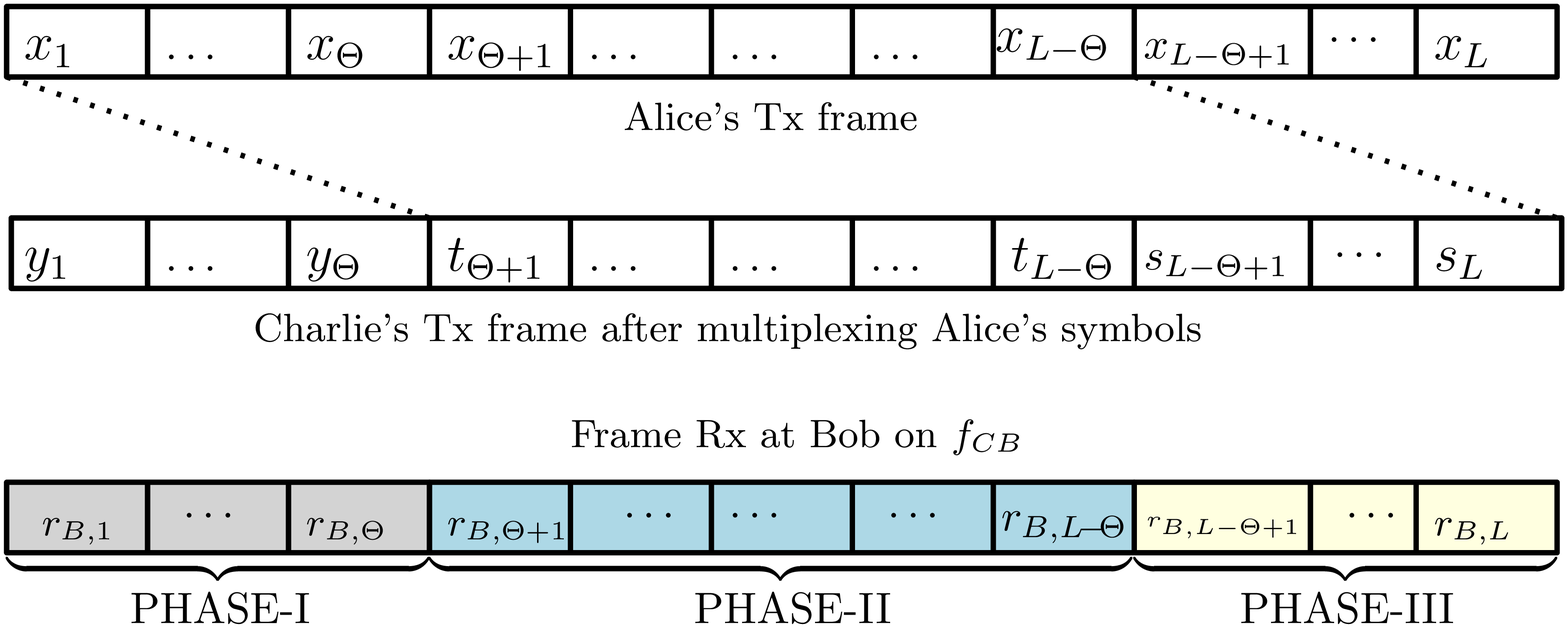}
\caption{\label{fig:frame_model} Frame model for the $3\phi$ DASC-MF scheme.}
    \end{minipage}%
    \hfill
   \begin{minipage}[b]{0.48\textwidth}
   \captionsetup{type=table,font=scriptsize}
    \caption{\label{tab:three_phase} SYMBOLS TRANSMITTED IN EACH PHASE}
\begin{center}
\scalebox{0.75}{
\begin{tabular}{|c|c|c|}
\hline
          & Alice-to-Bob link & Charlie-to-Bob link
\\ \hline
Phase-I  & $\sqrt{1-\alpha}x_{n}$  & $\sqrt{\alpha}y_{n}$
\\ \hline
\multirow{2}{*}{Phase-II}  & $\sqrt{1-\alpha}x_{n}$ & $t_{n}=
\begin{cases} 
y_{n} &\text{if } \hat{x}_{n-\Theta}=0,\\
\sqrt{\alpha}e^{\iota\frac{\pi}{M}}y_{n} &\text{if } \hat{x}_{n-\Theta}=1.
\end{cases}
$ 
\\ \hline
\multirow{2}{*}{Phase-III}  & $\sqrt{1-\beta}x_{n}$ & $s_{n}=
\begin{cases} 
\sqrt{\beta}y_{n} &\text{if } \hat{x}_{n-\Theta}=0,\\
\sqrt{\beta}e^{\iota\frac{\pi}{M}}y_{n} &\text{if } \hat{x}_{n-\Theta}=1.
\end{cases}
$ 
\\ \hline
\end{tabular}
}
\end{center}
    \end{minipage}%
    \end{center}
    \vspace{-0.5cm}
\end{figure*} 
 
From the above discussions, it is clear that, the decoding complexities during Phase-I, Phase-II and Phase-III are $\mathcal{O}(M)$, $\mathcal{O}(2M)$, and $\mathcal{O}(4M)$, respectively. Hence, unlike DASC-MF scheme, the worst-case complexity of $3\phi$ DASC-MF is linear in $M$. Further, since $\alpha$ is close to $1$, for $1\leq n\leq L-\Theta$, the energy level on $f_{CB}$ is solely controlled by Charlie and only the last $\Theta$ symbols  out of the $L$ symbols received at Bob are uncoordinated in energy. In particular, in $3\phi$ DASC-MF scheme, $\frac{\Theta}{L}$ fraction of symbols are uncoordinated in energy, whereas, in traditional DASC-MF scheme, the entire frame of $L$ symbols are uncoordinated in energy. Hence, $3\phi$ DASC-MF scheme helps reduce the decoding complexity and reduce the fraction of symbols over which the energy is uncoordinated.

\begin{figure}
\centering
\includegraphics[scale=0.33]{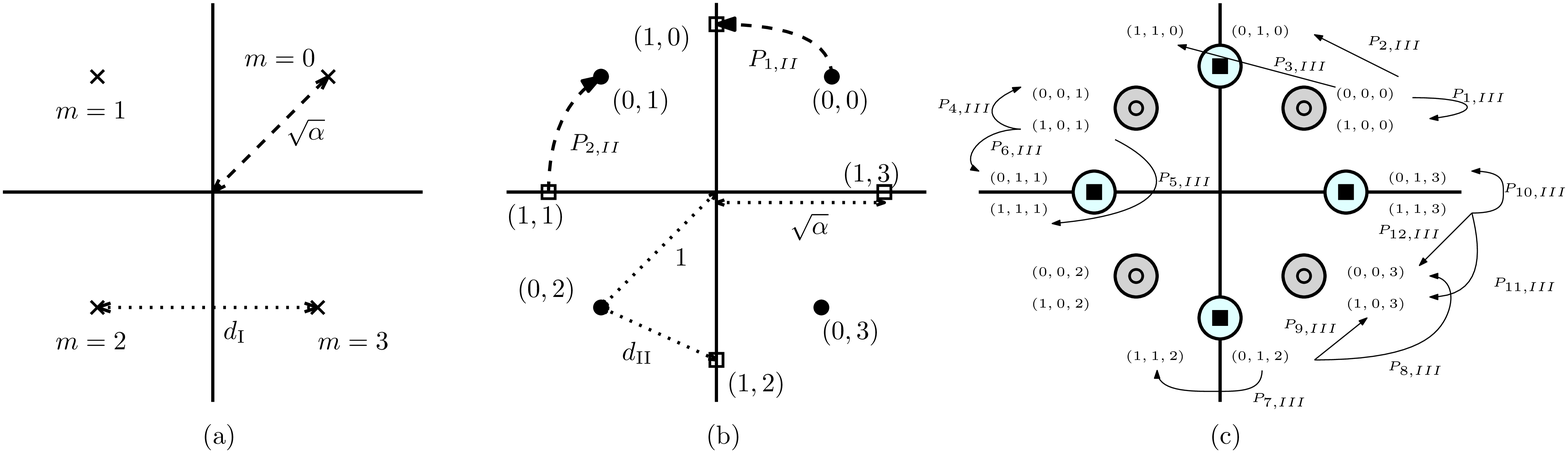}
\caption{\label{fig:cons} Constellation diagram depicting the symbols jointly contributed by Alice and Charlie during each phase in the $3\phi$ DASC-MF as a function of $x_{n}$, $x_{n-\Theta}$, and $y_{n}$: (a) Phase-I (b) Phase-II (c) Phase-III.}
\vspace{-0.5cm}
\end{figure}

Towards characterising the error performance of the $3\phi$ DASC-MF scheme, if $P_{e,3\phi}$ denotes the average probability of error across all the three phases, then
\bieee
P_{e,3\phi} = \frac{\Theta}{L}P_{I,avg} + \frac{(L-2\Theta)}{L}P_{II,avg} + \frac{\Theta}{L}P_{III,avg},\label{eq:Pe_main}
\eieee
\noindent where $P_{I,avg}$ denotes the average probability of error in decoding $y_{n}$, $1\leq n\leq \Theta$, during Phase-I, $P_{II,avg}$ denotes the average probability of error in jointly decoding $y_{n}$ and $x_{n-\Theta}$, $\Theta+1\leq n\leq L-\Theta$, during Phase-II, and $P_{III,avg}$ denotes the average probability of error in jointly decoding $x_{n}$, $y_{n}$, $x_{n-\Theta}$, $L-\Theta+1\leq n\leq L$, during Phase-III. All the error terms in \eqref{eq:Pe_main} are averaged over $h_{CB,n}$, $1\leq n\leq L$. Note that, here, $P_{I,avg}$ is a function of $\alpha$, $P_{II,avg}$ is a function of $N_{C}$ and $\alpha$, and $P_{III,avg}$ is a function of $N_{C}$ and $\beta$. Since $\alpha = 1-\Delta N_{o}$ is fixed for a given choice of $\Delta$ and SNR, the parameters of interest are $N_{C}$ and $\beta$. Therefore, in \eqref{opt} we formulate an optimization problem to compute the optimal values of $N_{C}$ and $\beta$ for a given $L$, $\Delta$ and SNR. 
\begin{mdframed}
\bieee
N_{C}^{\star},\beta^{\star} = \arg\underset{N_{C}, \beta}{\min} \quad &  & P_{e,3\phi};\ \ \text{s.t.: } N_{C}>1, 0<\beta < 1.\label{opt}
\eieee 
\end{mdframed}
\noindent In the subsequent sections, we discuss the signal model for all the three phases in detail and compute $P_{I,avg}$, $P_{II,avg}$, and $P_{III,avg}$, so as to obtain $P_{e,3\phi}$.

\subsection{Signal Model and Error Analysis of Phase-I}

During Phase-I, Bob receives the first $\Theta$ symbols of Alice on Alice-to-Bob link of the MAC scaled by the factor $\sqrt{1-\alpha}$ and the unmultiplexed symbols of Charlie scaled by the factor $\sqrt{\alpha}$ on Charlie-to-Bob link of the MAC. However, according to the protocol of $3\phi$ DASC-MF scheme, Bob treats the incoming symbols from Alice on the Alice-to-Bob link of the MAC as interference. Thus, the $n^{th}$ baseband symbol received at Bob during Phase-I is given as $r_{B,n} = h_{CB,n}\sqrt{\alpha}y_{n} + \tilde{w}_{B,n}$, for $1\leq n\leq \Theta$, where $\tilde{w}_{B,n} = \sqrt{1-\alpha}h_{AB,n}x_{n} + w_{B,n}$ denotes the effective noise at Bob, such that $N_{o\alpha} = N_{o}+1-\alpha$ is its worst case effective variance. Given that $\alpha$ is close to $1$, we assume that $\tilde{w}_{B,n}\sim\mathcal{CN}(0,N_{o\alpha})$. For $4-$PSK signalling scheme, the constellation diagram during Phase-I at Bob is as shown in Fig.~\ref{fig:cons}~(a), where $m$ denotes the index of the PSK symbol transmitted by Charlie. Owing to the Gaussian statistics of $\tilde{w}_{B,n}$ and $h_{CB,n}$, $r_{B,n}\vert_{y_{n}, h_{CB,n}}\sim\mathcal{CN}\left(\sqrt{\alpha}h_{CB,n}y_{n}, N_{o\alpha}\right)$. Since the symbols received at Bob during Phase-I only contains Charlie's symbols, based on $r_{B,n}$, the Maximum A Posteriori (MAP) decoder for Phase-I is
\bieee
\hat{m} = \arg\underset{m}{\max\ }g_{I}\left(r_{B,n}\vert y_{n} = e^{\iota\frac{\pi}{M}\left(2m+1\right)}, h_{CB,n}\right),\label{eq:MAP_P1}
\eieee
\noindent where $m\in\{0,\ldots,M-1\}$ and $g_{I}(\cdot)$ is the PDF of $r_{B,n}$ conditioned on $h_{CB,n}$ and $y_{n}$. Further, $\hat{m}$ denotes the decoded PSK index. Using \eqref{eq:MAP_P1}, in the following theorem, we characterise the average probability of error over all the realizations of $h_{CB,n}$, denoted by $P_{1,avg}$.

\begin{theorem}
\label{th:upper_phase1}
At high SNR, a union bound on the probability of error during Phase-I is approximated as $2Q\left(\frac{|h_{CB,n}| d_{\mathrm{I}}}{\sqrt{2N_{o\alpha}}}\right)$, where $Q(\cdot)$ denotes the Q-function. Further, we first use Chernoff bound to upper bound the error expression and then average it over all the realizations of $h_{CB,n}$ to obtain $\mathcal{P}_{1}$, such that $\mathcal{P}_{I} = \dfrac{4N_{o\alpha}}{4N_{o\alpha} + d_{\mathrm{I}}^{2}}$. \cite{dcom}
\end{theorem}

\subsection{Signal Model and Error Analysis of Phase-II}
 
During Phase-II, Alice continues to transmit her OOK symbol, $x_{n}$, with $1-\alpha$ fraction of her energy and Charlie transmits the multiplexed symbol $t_{n}$, which is a function of $y_{n}$ and $x_{n-\Theta}$. Thus, Bob receives Alice's symbols on the Alice-to-Bob link of the MAC and Charlie's multiplexed symbols on the Charlie-to-Bob link of the MAC. However, he considers Alice's symbols as interference and jointly decodes $y_{n}$ and $x_{n-\Theta}$. The $n^{th}$ baseband symbol received during Phase-II at Bob is given as $r_{B,n} = h_{CB,n}t_{n} + \tilde{w}_{B,n}$, for $\Theta+1\leq n\leq L-\Theta$, where $t_{n}$ is as given in Table~\ref{tab:three_phase} and $\tilde{w}_{B,n}$ is the effective noise at Bob. Fig.~\ref{fig:cons}~(b) depicts the constellation received  at Bob during Phase-II. If Alice and Charlie choose to transmit $x_{n-\Theta} = j$ and $y_{n} = e^{\iota\frac{\pi}{M}(2m+1)}$, then the transmit pair is denoted by $(j,m)$. It can be observed that for $4-$PSK used by Charlie, Bob can receive one out of the $8$ symbols. In general, when Charlie uses $M-$PSK constellation, Bob receives one out of the $2M$ constellation points. Therefore, the distribution of $r_{B,n}$ conditioned on $x_{n-\Theta}$, $y_{n}$, and $h_{CB,n}$ is
\begin{subnumcases}{}
\mathcal{CN}(h_{CB,n}y_{n}, N_{o\alpha}) & if $x_{n-\Theta} = 0$ and $\hat{x}_{n-\Theta} = x_{n-\Theta}$ or $x_{n-\Theta} = 1$ and $\hat{x}_{n-\Theta} \neq x_{n-\Theta}$,
\label{eq:rb_Ph2a}\\
\mathcal{CN}(\sqrt{\alpha}h_{CB,n}e^{\iota\frac{\pi}{M}}y_{n}, N_{o\alpha}) & if $x_{n-\Theta} = 0$ and $\hat{x}_{n-\Theta} \neq x_{n-\Theta}$ or $x_{n-\Theta} = 1$ and $\hat{x}_{n-\Theta} = x_{n-\Theta}$.
\label{eq:rb_Ph2d}
\end{subnumcases}
Using \eqref{eq:rb_Ph2a}~--~\eqref{eq:rb_Ph2d}, the joint MAP decoder for Phase-II is given by
\bieee
\hat{j},\hat{m}= \arg\underset{j,m}{\max\ }g_{II}\left(r_{B,n}\vert x_{n-\Theta}=j, y_{n} = e^{\iota\frac{\pi}{M}\left(2m+1\right)}, h_{CB,n}\right),\label{eq:MAP_p2}
\eieee
\noindent where $j\!\in\!\{0,1\}$ and $m\!\in\!\{0,\ldots, M-1\}$. Further, $g_{II}(\cdot)$ is the conditional PDF of $r_{B,n}$ conditioned on $x_{n-\Theta}$, $y_{n}$, and  $h_{CB,n}$. Note that due to errors introduced by Charlie in decoding Alice's symbols, the distribution of $g_{II}(\cdot)$ is a Gaussian mixture. In particular, the distribution of $g_{II}(\cdot)$ is a convex combination of $g_{II}\left(r_{B,n}\vert y_{n}, h_{CB,n}, x_{n-\Theta}= \hat{x}_{n-\Theta}\right)$ and $g_{II}\left(r_{B,n}\vert e^{\iota\frac{\pi}{M}}y_{n} , h_{CB,n}, x_{n-\Theta}\neq \hat{x}_{n-\Theta}\right)$, when $x_{n-\Theta} = 0$ and $x_{n-\Theta} = 1$. The distribution of $g_{II}(\cdot)$ for different realizations of $x_{n-\Theta}$ is as given in \eqref{eq:GM}. \begin{multline}
g_{II}\left(r_{B,n}\vert x_{n-\Theta}=j, y_{n}, h_{CB,n}\right)= \Phi_{jj} g_{II}\left(r_{B,n}\vert y_{n}, h_{CB,n}, x_{n-\Theta}= \hat{x}_{n-\Theta}\right) +\\
\ \Phi_{j\overline{j}} g_{II}\left(r_{B,n}\vert e^{\iota\frac{\pi}{M}}y_{n} , h_{CB,n}, x_{n-\Theta}\neq \hat{x}_{n-\Theta}\right),\label{eq:GM}
\end{multline}
\noindent where $\overline{j}$ is the complement of $j$ and $\Phi_{00} = 1-\Phi_{01}$ and $\Phi_{11} = 1-\Phi_{10}$ are the probabilities of correct detection of bit-0 and bit-1 at Charlie, respectively. Further, $g_{II}\left(r_{B,n}\vert e^{\iota\frac{\pi}{M}}y_{n} , h_{CB,n}, x_{n-\Theta}\neq \hat{x}_{n-\Theta}\right)$ is the conditional PDF of the symbol received at Bob when Charlie makes an error in decoding Alice's symbol. Since solving the error-performance of the joint MAP decoder using the Gaussian mixtures is non-trivial, we propose an approximation on \eqref{eq:MAP_p2}, where we only consider the dominant term in \eqref{eq:GM} for error computation, for each realization of $j$. Along the similar lines of \cite{TCCN-VH}, we observe that $\Phi_{00}$ is dominant over $\Phi_{01}$ when $x_{n-\Theta} = 0$ and $\Phi_{11}$ is dominant over $\Phi_{10}$, when $x_{n-\Theta} = 1$. Therefore, we approximate the joint MAP decoder in \eqref{eq:MAP_p2} by a Joint Dominant Decoder (JDD), by retaining the first term in the RHS of \eqref{eq:GM}. The expression of JDD is given as
\bieee
\hat{j},\hat{m}= \arg\underset{j,m}{\max\ }\tilde{g}_{II}\left(r_{B,n}\vert x_{n-\Theta}=j, y_{n} = e^{\iota\frac{\pi}{M}\left(2m+1\right)}, h_{CB,n}\right),\label{eq:JDD_p2}
\eieee
\noindent where $j\in\{0,1\}$ and $m\in\{0,\ldots,M-1\}$. Further, $\tilde{g}_{II}(\cdot)$ is an approximation on $g_{II}(\cdot)$ when considering the first term in the RHS of \eqref{eq:GM}. In the next theorem, we derive a union bound on the probability of error in jointly decoding Alice's and Charlie's symbols during Phase-II.

\begin{theorem}
\label{th:P2_3phase}
At high SNR, the probability of error for Phase-II, denoted by $P_{e,II}$, is upper bounded by $\Phi_{00}P_{1,II}  + \Phi_{01}P_{\overline{1},II} + \Phi_{11}P_{2,II} + \Phi_{10}P_{\overline{2},II}$, where $P_{1,II} = P_{2,II}$  and $P_{\overline{1},II}=P_{\overline{2},II} = 1-P_{1,II}$, such that $P_{1,II} \approx Q\left(\frac{|h_{CB,n}| d_{\mathrm{II}}}{\sqrt{2N_{o\alpha}}}\right)$. Here,  $d_{\mathrm{II}} = \sqrt{1 + \alpha - 2\sqrt{\alpha}\cos\left(\frac{\pi}{M}\right)}$ denotes the minimum Euclidean distance between the constellation points transmitted by Charlie during Phase-II.
\end{theorem}
\begin{proof}
Let a transmit pair corresponding to $x_{n-\Theta}$ and $y_{n}$ is denoted by $(j,m)$, such that $y_{n}$ denotes the PSK symbol transmitted by Charlie corresponding to the index $m$. Let $\nabla_{(j,m)\rightarrow(j^{\prime},m^{\prime})}$ be the event $(j',m')\neq (j,m)$, then a pair $(j,m)$ is incorrectly decoded as $(j',m')$ if
\bieee
\nabla_{(j,m)\rightarrow(j^{\prime},m^{\prime})}\triangleq \frac{\tilde{g}_{II}\left(r_{B,n}\vert x_{n-\Theta}=j, y_{n} = e^{\iota\frac{\pi}{M}\left(2m+1\right)}, h_{CB,n}\right)}{\tilde{g}_{II}\left(r_{B,n}\vert x_{n-\Theta}=j^{\prime}, y_{n} = e^{\iota\frac{\pi}{M}\left(2m^{\prime}+1\right)}, h_{CB,n}\right)}\leq 1.\nn
\eieee
\noindent Therefore, the probabilities of decoding a pair $(0,m)$ as $(j',m')$ and $(1,m)$ as $(j',m')$ at Bob is
\begin{small}
\bieee
\Pr\left((0,m)\rightarrow(j',m')\right) &=& \Phi_{00}\!\Pr\left(\!\nabla_{(0,m)\rightarrow(j',m')}\leq 1\vert \hat{x}_{n-\Theta}=x_{n-\Theta}\right)\!+\! \Phi_{01}\Pr\left(\nabla_{(0,m)\rightarrow(j',m')}\leq 1\vert \hat{x}_{n-\Theta}\neq x_{n-\Theta}\!\right),\label{eq:phase-2_1}\\
\Pr\left((1,m)\rightarrow(j',m')\right) &=& \Phi_{11}\!\Pr\!\left(\nabla_{(1,m)\rightarrow(j',m')}\leq 1\vert \hat{x}_{n-\Theta}=x_{n-\Theta}\right) \!+\! \Phi_{10}\Pr\left(\nabla_{(1,m)\rightarrow(j',m')}\leq 1\vert \hat{x}_{n-\Theta}\neq x_{n-\Theta}\!\right).\label{eq:phase-2_2}
\eieee
\end{small}
Combining \eqref{eq:phase-2_1} and \eqref{eq:phase-2_2}, the probability of error in decoding a transmit pair $(j,m)$ as $(j',m')$ is 
\begin{small}
\bieee
\Pr((j,m)\rightarrow (j',m')) = \Phi_{jj}\Pr\left(\nabla_{(j,m)\rightarrow(j',m')}\leq 1\vert \hat{x}_{n-\Theta}=x_{n-\Theta}\right) + \Phi_{j\overline{j}}\Pr\left(\nabla_{(j,m)\rightarrow(j',m')}\leq 1\vert \hat{x}_{n-\Theta}\neq x_{n-\Theta}\right).\label{eq:phase2_error3}
\eieee
\end{small}
\noindent Therefore, if the overall probability of error in decoding a transmit pair $(j,m)$ is denoted by $\Pr\left((\hat{j}, \hat{m})\neq (j,m)\right)$, then an upper bound on the overall expression is
\bieee
\Pr\left((\hat{j},\hat{m}\neq (j,m)\right)\leq\underset{(j^{\prime},m^{\prime})\neq (j,m)}{\sum_{j^{\prime}=0}^{1}\sum_{m^{\prime}=0}^{M-1}}\Pr((j,m)\rightarrow (j',m')).\label{eq:phase2_error1}
\eieee
\noindent As discussed earlier, Bob receives one out of the $2M-$PSK  symbols  corresponding to the $M-$PSK constellation used by Charlie. Finally, if $P_{e,II}$ denotes the average probability of error over $2M$ symbols, then using \eqref{eq:phase-2_1}~-~\eqref{eq:phase2_error1}, it is straightforward to compute an upper bound on $P_{e,II}$ as $P_{e,II} \leq \frac{1}{2M}\sum_{j=0}^{1}\sum_{m=0}^{M-1}M \Pr\left((\hat{j},\hat{m})\neq (j,m)\right)\approx\Phi_{00}P_{1,II}  + \Phi_{01}P_{\overline{1},II} + \Phi_{11}P_{2,II} + \Phi_{10}P_{\overline{2},II}$. Here, we only considered the nearest-neighbours  at high SNR to approximate the upper bound. The pairwise probability terms in $P_{e,II}$ are depicted in Fig.~\ref{fig:cons}~(b). When evaluating the  various error terms in the upper bound, we observe that $P_{1,II} = P_{2,II}$  and $P_{\overline{1},II}=P_{\overline{2},II} = 1-P_{1,II}$, such that $P_{1,II} \approx Q\left(\frac{|h_{CB,n}| d_{\mathrm{II}}}{\sqrt{2N_{o\alpha}}}\right)$. Here, $d_{\mathrm{II}} = \sqrt{1 + \alpha - 2\sqrt{\alpha}\cos\left(\frac{\pi}{M}\right)}$ denotes the minimum Euclidean distance between the constellation points transmitted by Charlie during Phase-II.  
\end{proof}
\begin{corollary}
\label{col:upper_Phase2}
If $P_{II,avg} = \mathbb{E}_{h_{CB,n}}[P_{e,II}]$ denotes the average probability of error over all the realizations of $h_{CB,n}$, then, using the Chernoff bound $P_{II,avg} \leq \mathcal{P}_{II} \triangleq \dfrac{2N_{o\alpha}}{4N_{o\alpha} + d_{\mathrm{II}}^{2}}\left(\Phi_{00} + \Phi_{11}\right) + \left(\Phi_{01}+\Phi_{10}\right),$ where we have used $P_{\overline{1},II}\leq 1$, $P_{\overline{2},II}\leq 1$ in the upper bound given in Theorem~\ref{th:P2_3phase}.
\end{corollary}

\subsection{Signal Model and Error Analysis of Phase-III}

During Phase-III, Bob observes the symbols on both the links, i.e., the Alice-to-Bob link of the MAC and the Charlie-to-Bob of the MAC to decode $x_{n}$, $y_{n}$, and $x_{n-\Theta}$, for $L-\Theta+1\leq n\leq L$. Unlike Phase-I and Phase-II, Alice and Charlie use the energy-splitting factor, $\beta\in(0,1)$ to transmit their respective symbols. In particular, Charlie scales his PSK symbols with $\sqrt{\beta}$, irrespective of $\hat{x}_{n-\Theta}=0$ or $\hat{x}_{n-\Theta}=1$. Further, Alice transmits her OOK symbols with $1-\beta$ fraction of her energy. With this modification, the $n^{th}$ baseband symbol received at Bob during Phase-III is given as
\bieee
r_{B,n} = h_{AB,n}\sqrt{1-\beta}x_{n} + h_{CB,n}\sqrt{\beta}s_{n} + {w}_{B,n},\hspace{0.5cm} L-\Theta+1\leq n\leq L,\label{eq:rb_p3}
\eieee
\noindent where $s_{n}$ is defined in Table~\ref{tab:three_phase} and $w_{B,n}\sim\mathcal{CN}(0,N_{o})$ is the AWGN at Bob. Owing to the non-coherent nature of Alice-to-Bob link of the MAC, the variance of the effective noise is $N_{o}+1-\beta$ and $N_{o}$, when $x_{n}=1$ and $x_{n}=0$, respectively. Fig.~\ref{fig:cons}~(c) depicts the received constellation symbols at Bob during Phase-III, when Charlie uses $4-$PSK signalling. If Alice's current symbol, $x_{n}=i$, Alice's multiplexed symbol, $x_{n-\Theta} = j$, and Charlie's current symbol, $y_{n} = e^{\iota\frac{\pi}{M}(2m+1)}$, then a transmit triplet is denoted by $(i,j,m)$. The variance of the effective noise corresponding to $i=0$ is $N_{o}$ (small disk and solid squares), whereas, the variance of the effective noise corresponding to $i=1$ is $N_{o}+1-\beta$ (blue and grey disks). We highlight that, Bob receives one out of the $16-$PSK symbols corresponding to $4-$PSK used by Charlie. In general, when Charlie uses $M-$PSK constellation, Bob is likely to receive one out of the $4M-$PSK symbols, out of which, $2M-$ PSK symbols are received with variance $N_{o}$ and the rest $2M-$ PSK symbols are received with variance $N_{o}+1-\beta$. The distribution of $r_{B,n}$ as a function of $x_{n}$, $x_{n-\Theta}$,  $y_{n}$, and $h_{CB,n}$ is $r_{B,n}\vert_{x_{n},x_{n-\Theta},y_{n},h_{CB,n}}\sim\mathcal{CN}(\sqrt{\beta}h_{CB,n}s_{n}, N_{o})$, if $x_{n}=0$ and $r_{B,n}\vert_{x_{n},x_{n-\Theta},y_{n},h_{CB,n}}\sim\mathcal{CN}(\sqrt{\beta}h_{CB,n}s_{n}, N_{o\beta})$, if $x_{n}=1$. Thus, using the distribution of $r_{B,n}$, the joint MAP decoder for Phase-III is
\bieee
\hat{i},\hat{j},\hat{m}= \arg\underset{\substack{i,j,m}}{\max\ }g_{III}\left(r_{B,n}\vert x_{n}=i,x_{n-\Theta}=j, y_{n}= e^{\iota\frac{\pi}{M}\left(2m+1\right)}, h_{CB,n}\right),\label{eq:MAP_p3}
\eieee
\noindent where $i\in\{0,1\}$, $j\in\{0,1\}$, and $m\in\{0,\ldots, M-1\}$. Further, $g_{III}(\cdot)$ is the conditional PDF of $r_{B,n}$, conditioned on $x_{n-\Theta}$, $x_{n}$, $y_{n}$, and $h_{CB,n}$. Along the similar lines of Phase-II, we note that $g_{III}(\cdot)$ in \eqref{eq:MAP_p3} is also a Gaussian mixture for various realizations of $x_{n-\Theta}$. Therefore, we approximate $g_{III}(\cdot)$ as $\tilde{g}_{III}(\cdot)$, where we retain the dominant terms from the Gaussian mixture, for $x_{n-\Theta} = 0$ and $x_{n-\Theta}=1$, respectively. The JDD for Phase-III of $3\phi$ DASC-MF is
\bieee
\hat{i},\hat{j},\hat{m}= \arg\underset{i,j,m}{\max\ }\tilde{g}_{III}\left(r_{B,n}\vert x_{n}=i,x_{n-\Theta}=j, y_{n}= e^{\iota\frac{\pi}{M}\left(2m+1\right)}, h_{CB,n}\right),\label{eq:JDD_p3}
\eieee
\noindent  where $i,j\in\{0,1\}$ and $m\in\{0,\ldots, M-1\}$. In the next theorem, we compute a union bound on the probability of error in jointly decoding $x_{n}$, $x_{n-\Theta}$, and $y_{n}$, when using the JDD presented in \eqref{eq:JDD_p3}.
\begin{table}[!htb]
    \caption{\label{tab:error_tab}ERROR TERMS FOR PHASE-III OF $3\phi$ DASC-MF SCHEME AS GIVEN IN THEOREM~\ref{th:P3_3phase}}
      \centering
      \scalebox{0.75}{
         \begin{tabular}{|c l||c l||c l|} 
  \hline
$P_{1,III}=$ & $e^{\frac{-\varrho}{N_{o}}}$ & $P_{2,III}=$ & $Q\left(\frac{|h_{CB,n}| d_{\mathrm{III}}}{\sqrt{2N_{o}}}\right)$ & $P_{3,III}=$ & $Q_{1}\left(\frac{|A|}{\sqrt{N_{o}/2}}, \frac{\sqrt{\xi_{1}}}{\sqrt{N_{o}/2}}\right)$
\\ \hline 
$P_{4,III}=$ & $ 1 - e^{\frac{-\varrho}{N_{o\beta}}}$ & $P_{5,III}=$ & $Q\left(\frac{|h_{CB,n}| d_{\mathrm{III}}}{\sqrt{2N_{o\beta}}}\right)$ & $P_{6,III}=$ & $1 - Q_{1}\left(\frac{|B|}{\sqrt{N_{o\beta}/2}}, \frac{\sqrt{\xi_{2}}}{\sqrt{N_{o\beta}/2}}\right)$
\\ \hline
$P_{9,III}=$ & $Q_{1}\left(\frac{|A|}{\sqrt{N_{o}/2}}, \frac{\sqrt{\xi_{2}}}{\sqrt{N_{o}/2}}\right)$ & $P_{12,III}=$ & $1 - Q_{1}\left(\frac{|B|}{\sqrt{N_{o\beta}/2}}, \frac{\sqrt{\xi_{1}}}{\sqrt{N_{o\beta}/2}}\right)$ & $P_{\overline{3},III} =$ & $Q_{1}\left(\frac{|B|}{\sqrt{N_{o}/2}}, \frac{\sqrt{\xi_{1}}}{\sqrt{N_{o}/2}}\right)$
\\ \hline
$P_{\overline{6},III} =$ & $1 - Q_{1}\left(\frac{|A|}{\sqrt{N_{o\beta}/2}}, \frac{\sqrt{\xi_{2}}}{\sqrt{N_{o\beta}/2}}\right)$ & $P_{\overline{9},III} =$ & $Q_{1}\left(\frac{|B|}{\sqrt{N_{o}/2}}, \frac{\sqrt{\xi_{2}}}{\sqrt{N_{o}/2}}\right)$ & $P_{\overline{12},III} =$ & $1 - Q_{1}\left(\frac{|A|}{\sqrt{N_{o\beta}/2}}, \frac{\sqrt{\xi_{1}}}{\sqrt{N_{o\beta}/2}}\right)$
\\ \hline
   \end{tabular}
   } 
   \vspace{-0.5cm}
\end{table}

\begin{theorem}
\label{th:P3_3phase}
At high SNR, the error probability for Phase-III, denoted by $P_{e,III}$, is upper bounded by
\begin{small}
\begin{multline}
\frac{1}{4}\Big[\!\Phi_{00}\!\left(\!P_{1,III}\!\!+\! \!2P_{2,III}\!\!+\!\! 2P_{3,III}\!\! + P_{4,III}\!\! + \!\! 2P_{5,III}\!\! +\!\! 2P_{6,III}\!\right)\! +\! \Phi_{01}\!\left(\!P_{\overline{1},III}\!\!  +\!\! 2P_{\overline{2},III}\!\!  +\!\! 2P_{\overline{3},III}\!\! +\!\! P_{\overline{4},III}\!\! +\!\!2P_{\overline{5},III}\!\! +\!\! 2P_{\overline{6},III} \!\right) + \\
\ \Phi_{11}\!\!\left(\!P_{7,III}\!\! +\!\! 2P_{8,III}\!\! +\!\! 2P_{9,III}\!\! +\!\! P_{10,III}\!\! +\!\! 2P_{11,III}\!\! +\!\! 2P_{12,III}\right)\!\! +\!\! \Phi_{10}\!\!\left(\!P_{\overline{7},III}\!\! +\!\! 2P_{\overline{8},III}\!\! +\!\! 2P_{\overline{9},III}  \!\!+\!\! P_{\overline{10},III}\!\! +\!\! 2P_{\overline{11},III}\!\! +\!\! 2P_{\overline{12},III} \!\right)\!\!\Big],\label{eq:phase3_error3}
\end{multline}
\end{small}
\noindent where the various error terms in \eqref{eq:phase3_error3} are tabulated in Table~\ref{tab:error_tab}, such that, $P_{4,III} = P_{10,III} = P_{\overline{4},III}   = P_{\overline{10},III} $; $P_{5,III}= P_{11,III} = 1-P_{\overline{5},III} = 1-P_{\overline{11},III} $; $P_{1,III} = P_{7,III}=P_{\overline{1},III} =P_{\overline{7},III} $; $P_{2,III} = P_{8,III} = 1-P_{\overline{2},III}  =  1-P_{\overline{8},III}$.
\end{theorem}

\begin{proof}
This can be proved along the similar lines of Theorem~\ref{th:P2_3phase}.
\end{proof}
In addition to the terms defined in Table~\ref{tab:error_tab}, $\varrho = \frac{N_{o}N_{o\beta}}{N_{o}-N_{o\beta}}\ln\left(\frac{N_{o}}{N_{o\beta}}\right)$ is the threshold for non-coherent energy detection at Bob between $(0,j,m)$ and $(1,j,m)$, $A = \frac{|h_{CB,n}| d_{\mathrm{III}}N_{o}}{N_{o}-N_{o\beta}}$, $B = \frac{|h_{CB,n}| d_{\mathrm{III}}N_{o\beta}}{N_{o}-N_{o\beta}}$, $\xi_{1} = \frac{N_{o}N_{o\beta}}{N_{o}-N_{o\beta}}\left[\ln\left(\frac{N_{o}\Phi_{11}}{N_{o\beta}\Phi_{00}}\right) + \frac{|h_{CB,n}|^{2}d_{\mathrm{III}}^{2}}{N_{o}-N_{o\beta}}\right]$, and $\xi_{2} = \frac{N_{o}N_{o\beta}}{N_{o}-N_{o\beta}}\left[\ln\left(\frac{N_{o}\Phi_{00}}{N_{o\beta}\Phi_{11}}\right) + \frac{|h_{CB,n}|^{2}d_{\mathrm{III}}^{2}}{N_{o}-N_{o\beta}}\right]$ are the parameters of the respective Marcum-Q functions $(Q_{1}(\cdot,\cdot))$, where $d_{\mathrm{III}} = 2\sqrt{\beta}\sin\frac{\pi}{2M}$ is the minimum Euclidean distance between the constellation symbols received during Phase-III.

For $M=4$, the various error terms in \eqref{eq:phase3_error3} are depicted in Fig.~\ref{fig:cons}~(c). We note that the upper bound on $P_{e,III}$ in \eqref{eq:phase3_error3} contains exponential functions, Q-functions, and Marcum-Q functions. Although it is straightforward to compute the average of Q-functions over various realizations of $h_{CB,n}$ in closed-form, averaging Marcum-Q functions over various realizations of $h_{CB,n}$ is non-tractable for certain cases. Therefore, in the next lemma, we directly provide an upper bound on the Marcum-Q function averaged over the realizations of $h_{CB,n}$ to simplify the analysis.
\begin{lemma}
\label{lm:P3_P1}
The term $\mathbb{E}_{h_{CB,n}}\left[P_{3,III}\right]$ is upper bounded by $ P_{1,III}$, for all $0<\beta<1$.
\end{lemma}
\begin{proof}
Along the similar lines of \cite{TCCN-VH}, we compute $\mathbb{E}_{h_{CB,n}}\left[P_{3,III}\right] = \left(\frac{N_{o}\Phi_{11}}{N_{o\beta}\Phi_{00}}\right)^{\frac{N_{o\beta}}{N_{o\beta}-N_{o}}} \frac{\left(N_{o\beta}-N_{o}\right)^{2}}{\left(N_{o\beta}-N_{o}\right)^{2} + d_{\mathrm{III}}^{2}N_{o\beta}}$. From the expression of $\Phi_{00}$ and $\Phi_{11}$, we have $\Phi_{00}\geq \Phi_{11}$, thus, $\frac{\mathbb{E}_{h_{CB,n}}\left[P_{3,III}\right]}{P_{1,III}} \leq 1$.
\end{proof}
 In addition to Lemma~\ref{lm:P3_P1}, for moderate and high SNRs, we observe that, $\mathbb{E}_{h_{CB,n}}\left[P_{3,III}\right]\ll P_{1,III}$. Therefore, we use the upper bound $\mathbb{E}_{h_{CB,n}}[2P_{3,III}]\leq P_{1,III}$. Along the similar lines of Lemma~\ref{lm:P3_P1}, the following inequalities also hold good: $\mathbb{E}_{h_{CB,n}}[2P_{6,III}]\!\!\leq\!\! P_{4,III}$, $\mathbb{E}_{h_{CB,n}}[2P_{9,III}]\!\!\leq\!\! P_{7,III}$, and $\mathbb{E}_{h_{CB,n}}[2P_{12,III}]\leq P_{10,III}$. Furthermore, we also upper bound all the error terms that are the coefficients of $\Phi_{01}$ and $\Phi_{10}$ in \eqref{eq:phase3_error3} by $1$, thereby circumventing the non-tractable issue of Marcum-Q functions. 

\begin{corollary}
\label{col:upper_Phase3}
If $P_{III,avg} = \mathbb{E}_{h_{CB,n}}[P_{e,III}]$ denotes the average probability of error over all the realizations of $h_{CB,n}$, then an upper bound on $P_{III,avg}$ is given as $\mathcal{P}_{III}\!\!\triangleq\!\!\frac{1}{2}\!\left(\!\Phi_{00}\!\!\left(P_{1,III}\! +\! P_{2,III}^{\star}\! +\! P_{4,III}\! +\! P_{5,III}^{\star}\!\right)\!\right.+$ $\left. 5\left(\Phi_{01}+\Phi_{10}\right) + \Phi_{11}\left(P_{7,III}+P_{8,III}^{\star} + P_{10,III} + P_{11,III}^{\star}\right)\right)$, where $P_{2,III}^{\star} = P_{8,III}^{\star} = \frac{2N_{o}}{4N_{o} + d_{\mathrm{III}}^{2}}$, $P_{5,III}^{\star} = P_{11,III}^{\star}= \frac{2N_{o\beta}}{4N_{o\beta} + d_{\mathrm{III}}^{2}}$.
\end{corollary}

\subsection{Optimization of $N_{C}$ and $\beta$ for $3\phi$ DASC-MF Relaying Scheme}
\label{ssec:opt_beta_Nc}

Substituting $\mathcal{P}_{I}$, $\mathcal{P}_{III}$, and $\mathcal{P}_{III}$ from Theorem~\ref{th:upper_phase1}, Corollary~\ref{col:upper_Phase2}, and Corollary~\ref{col:upper_Phase3}, respectively in \eqref{eq:Pe_main}, we obtain an upper bound on the average probability of error of $3\phi$ DASC-MF, denoted by $\mathcal{P}_{e,3\phi}$ as
\bieee
P_{e,3\phi}\leq\mathcal{P}_{e,3\phi} \triangleq \frac{\Theta}{L}\mathcal{P}_{I} + \frac{(L-2\Theta)}{L}\mathcal{P}_{II} + \frac{\Theta}{L}\mathcal{P}_{III}.\label{eq:Pe_main2}
\eieee
\noindent Therefore, instead of solving \eqref{opt}, we solve an alternate optimization problem of minimising $\mathcal{P}_{e,3\phi}$ over the variables on interest, $N_{C}$ and $\beta$. Thus, the modified optimization problem is given as
\begin{mdframed}
\bieee
N_{C}^{\dagger},\beta^{\dagger} = \arg\underset{N_{C}, \beta}{\min} \quad &  & \mathcal{P}_{e,3\phi};\ \ \text{s.t.: } N_{C}>1, 0<\beta < 1.\label{opt_alt}
\eieee 
\end{mdframed}

Unlike Phase-I and Phase-II, decoding in Phase-III is a combination of coherent and non-coherent detection. As a result, for a given  $N_{C}$, $\beta$, and SNR, the error-rates for Phase-III dominates Phase-I and Phase-II, and thus, dominates $\mathcal{P}_{e,3\phi}$. Although, we can achieve improved error-rates during Phase-III by increasing $N_{C}$, we cannot indefinitely increase $N_{C}$, because $\Theta$ is an increasing function of $N_{C}$ and for large values of $N_{C}$, the fraction of symbols decoded during Phase-III increases, which in turn increases $\mathcal{P}_{e,3\phi}$. Therefore, we must use an appropriate $N_{C}$ that solves \eqref{opt_alt}. Towards solving \eqref{opt_alt}, we observe that proving unimodality of $\mathcal{P}_{e,3\phi}$ as a function of $N_{C}$ and $\beta$ is challenging due to the presence of the upper and the lower Gamma functions in the expressions of $\mathcal{P}_{II}$ and $\mathcal{P}_{III}$. Therefore, in this section, we first fix $N_{C}$ to analyse $\mathcal{P}_{e,3\phi}$ as a function of $\beta$ and then propose a low-complexity algorithm to obtain the near-optimal values of $N_{C}$ and $\beta$ that minimizes $\mathcal{P}_{e,3\phi}$.

Towards minimising $\mathcal{P}_{e,3\phi}$, we observe that when we fix $N_{C}$ and vary $\beta$, $\mathcal{P}_{e,3\phi}$ has a unique dip, for $\beta\in(0,1)$. This is due to the fact that, when we fix $\alpha = 1-\Delta N_{o}$ and $N_{C}$, $\mathcal{P}_{I}$ and $\mathcal{P}_{II}$ are independent of $\beta$, but $\mathcal{P}_{III}$ has a unique dip for $\beta\in(0,1)$. Further, we also observe that, the unique dip of $\mathcal{P}_{III}$ is close to the intersection of the increasing and decreasing terms of $\mathcal{P}_{III}$. The above observation is exemplified in Fig.~\ref{fig:inter} for $\Theta = \lceil 10\log_{10}(N_{C})\rceil$ at SNR = $25$ dB and $\Delta = 0.1$. Therefore, in the next lemma, we fix $N_{C}$ and identify the increasing and decreasing terms in $\mathcal{P}_{III}$ as a function of $\beta$. Subsequently, in Theorem~\ref{th:p3_dec}, we show that the increasing and decreasing terms in $\mathcal{P}_{III}$ intersect only once for $\beta\in(0,1)$.\looseness = -1
\begin{lemma}
\label{lm:inc_dec_P3}
For a fixed $\alpha$ and $N_{C}$, if $P_{+}$ and $P_{-}$ denote the increasing and decreasing terms in $\mathcal{P}_{III}$, respectively, w.r.t. $\beta$, then, $P_{+}\! =\! \Phi_{00}\!\left(P_{1,III}\! +\! P_{4,III}\right)\! +\! \Phi_{11}\!\left(P_{7,III}\! +\! P_{10,III}\right)$ and $P_{-}\! =\! \Phi_{00}\left(P_{2,III}^{\star}\! +\! P_{5,III}^{\star}\right) \!+\! \Phi_{11}\left(P_{8,III}^{\star}\! +\! P_{11,III}^{\star}\right)$.
\end{lemma}
\begin{proof}
Along the similar lines of Remark~\ref{rem:P01P10alpha}, we can prove that the terms $P_{1,III}$, $P_{4,III}$, $P_{7,III}$, and $P_{10,III}$ are increasing functions of $\beta$. Further, since $\Phi_{00}$ and $\Phi_{11}$ are independent of $\beta$, $P_{+}\! =\! \Phi_{00}\!\left(P_{1,III}\! +\! P_{4,III}\right)\! +\! \Phi_{11}\!\left(P_{7,III}\! +\! P_{10,III}\right)$ is an increasing function of $\beta$.

Further, the expressions of $P_{2,III}^{\star}$, $P_{5,III}^{\star}$, $P_{8,III}^{\star}$, and $P_{11,III}^{\star}$ are such that, $P_{2,III}^{\star} = P_{8,III}^{\star} = \frac{2N_{o}}{4N_{o} + d_{\mathrm{III}}^{2}}$ and $P_{5,III}^{\star} = P_{11,III}^{\star}= \frac{2N_{o\beta}}{4N_{o\beta} + d_{\mathrm{III}}^{2}}$. Differentiating $P_{2,III}^{\star}$ w.r.t. $\beta$ we get, $-\frac{2N_{o}}{(4N_{o} + d_{\mathrm{III}}^{2})^{2}}\sin^{2}\frac{\pi}{2M}$. Therefore, $P_{2,III}^{\star}$ and $P_{8,III}^{\star}$ are decreasing functions of $\beta$. Along similar lines, we can prove that $P_{5,III}^{\star}$ and $P_{11,III}^{\star}$ are decreasing functions of $\beta$. Thus, we have $P_{-} = \Phi_{00}\left(P_{2,III}^{\star} + P_{5,III}^{\star}\right) + \Phi_{11}\left(P_{8,III}^{\star} + P_{11,III}^{\star}\right)$.
\end{proof}

\begin{theorem}
\label{th:p3_dec}
 For a fixed $\alpha$ and $N_{C}$, $P_{+}$ and $P_{-}$ intersect only once for $\beta\in(0,1)$.
\end{theorem}
\begin{proof}
This can be proved along the similar lines of \cite[Theorem~3]{TCCN-VH}.
\end{proof}
\begin{rem}
\label{rem:intersection}
The unique intersection of $P_{+}$ and $P_{-}$ can be computed using the Newton-Raphson (NR) algorithm.
\end{rem}
Using the insights of Lemma~\ref{lm:inc_dec_P3} and Theorem~\ref{th:p3_dec}, we now present a low-complexity algorithm in Algorithm~\ref{Algo}, referred to as the $N_{C}$-$\beta$ Optimization algorithm, which provides a local minima of $\mathcal{P}_{e,3\phi}$ over the variables, $N_{C}$ and $\beta$.
\begin{figure*}
\vspace{-0.5cm}
\begin{minipage}[t]{.5\textwidth}
\captionsetup{width=0.8\textwidth}
\vspace{-4cm}
\centering
\includegraphics[scale = 0.5]{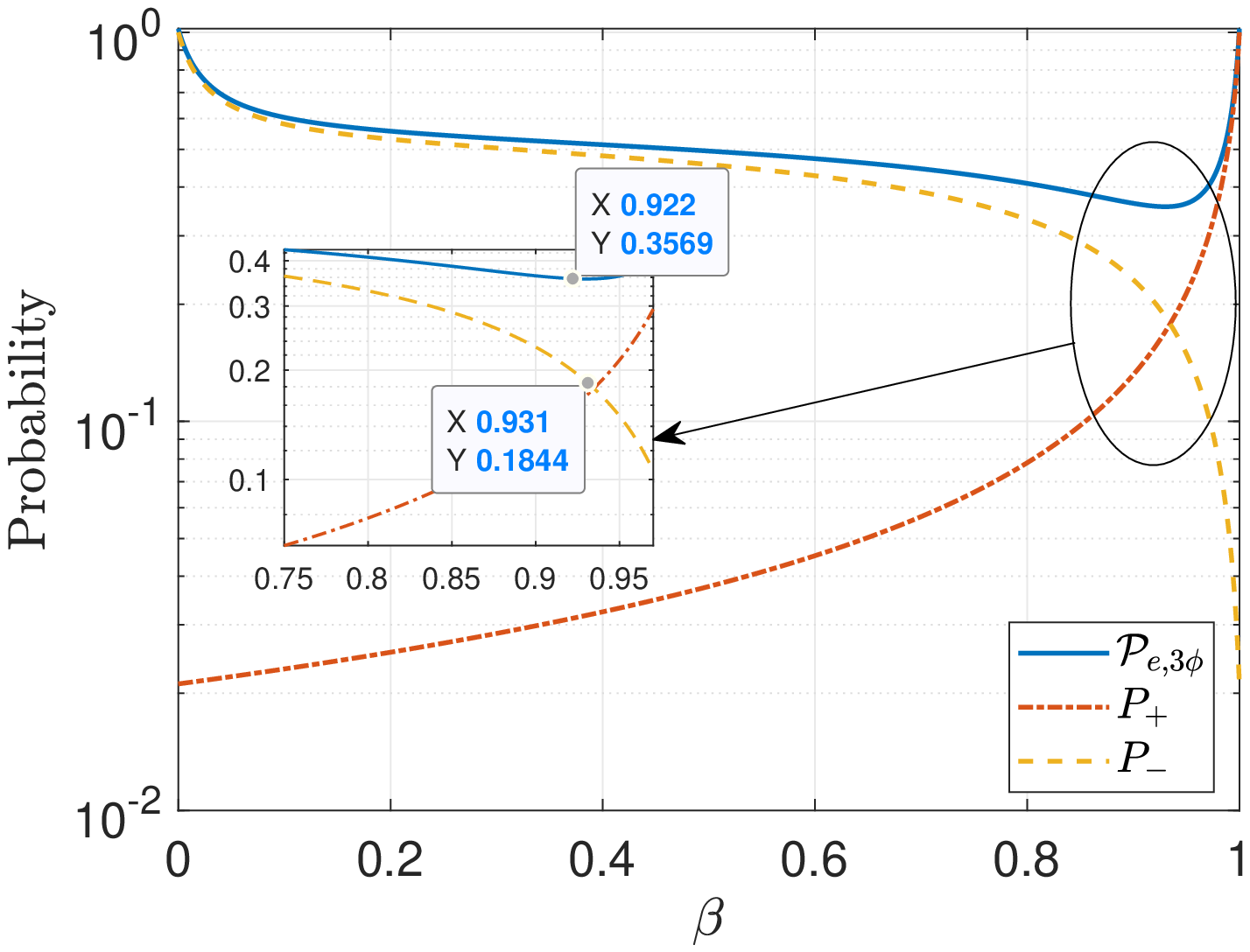}
\caption{\label{fig:inter}Variation of $\mathcal{P}_{e,3\phi}$ and its increasing and decreasing terms as a function of $\beta$ at $25$ dB.}
\end{minipage}%
\hfill
\begin{minipage}[t]{.5\textwidth}
\scalebox{0.8}{
\begin{algorithm}[H]
\KwInput{$\mathcal{P}_{e,3\phi}$, $P_{+}$, $P_{-}$, $\bigtriangleup$, $SNR$, $\delta_{tol}$}
\KwOutput{$N_{C}^{\dagger}$, $\beta^{\dagger}$}
$\beta\gets\beta_{o}$\;
$N_{C}\gets 1$\;
$P_{o} = 0.5$\;
\While{true}{
Use NR algorithm to compute intersection of $P_{+}$ and $P_{-}$, denoted by $\beta_{int}$\;
Substitute $\beta_{int}$ in $\mathcal{P}_{e,3\phi}$ and compute $P_{eval}$\;
  \eIf{$|P_{o}-P_{eval}|>\delta_{tol}$}{
      $N_{C}\gets N_{C}+1$\;
      $\beta_{o} \gets \beta_{int}$\;
      $P_{o}\gets P_{eval}$ 
  }{
  $N_{C}^{\dagger} = N_{C}$\;
  $\beta^{\dagger} = \beta_{int}$
  }
}
\caption{\label{Algo} $N_{C}$-$\beta$ Optimization}
\end{algorithm}
}
\end{minipage}
\vspace{-0.5cm}
\end{figure*}
We start the algorithm with $N_{C}=1$, an initial estimate of $\mathcal{P}_{e,3\phi}$, denoted by $P_{o}$, and an initial estimate of the intersection of $P_{+}$ and $P_{-}$, denoted by $\beta_{o}$. We use the NR algorithm to compute the intersection of $P_{+}$ and $P_{-}$ as $\beta_{int}$ and evaluate $P_{eval}$ by substituting $N_{C}$ and $\beta_{int}$ in $\mathcal{P}_{e,3\phi}$. We iteratively compute $\beta_{int}$ by incrementing $N_{C}$ in steps of $1$ until the absolute value of the difference between $P_{eval}$ and $P_{o}$ is less than the tolerance value, $\delta_{tol}$. When $\vert P_{o}-P_{eval}\vert<\delta_{tol}$, we exit the while-loop with the near-optimal values of $N_{C}$ and $\beta$, denoted by $N_{C}^{\dagger}$ and $\beta^{\dagger}$, respectively.  

\section{Semi-Coherent Multiple  Access Channel Scheme}
\label{sec:sc-mac}

A major limitation of the $3\phi$ DASC-MF scheme proposed in the previous section is that when $\Theta>\frac{L}{2}$, only a fraction of Alice's symbols can be recovered using the multiplexed symbols within the deadline and a majority of the multiplexed symbols are received at Bob after the deadline, resulting in violation of the latency constraint. As a result, when $\Theta>\frac{L}{2}$, Alice and Charlie must resort to uncoordinated multiple access communication, wherein Alice and Charlie transmit their symbols using $1-\varepsilon$ and $\varepsilon$ fractions of their energies, respectively, on $f_{CB}$ over the MAC, where $\varepsilon\in(0,1)$ is the design parameter under consideration. Here, Charlie does not decode Alice's symbol and only transmits his $M-$PSK symbols. Subsequently, Bob uses the received symbols to jointly decode Alice's and Charlie's symbols. A major distinction between the $3\phi$ DASC-MF scheme and SC-MAC is that due to no knowledge of Alice's symbols, Charlie always scales his constellation by energy $\varepsilon$. This phenomenon is similar to Phase-III $3\phi$ DASC-MF scheme. Furthermore, since the Alice-to-Bob link of the MAC is non-coherent and Charlie-to-Bob link of the MAC is coherent, we refer to this scheme as the Semi-Coherent Multiple Access Channel (SC-MAC). Subsequently, similar to the $3\phi$ DASC-MF scheme, Alice and Charlie transmit dummy OOK symbols from a Gold-sequence based scrambler with $\varepsilon$ and $1-\varepsilon$ fractions of their energies, respectively, on $f_{AB}$ to evade the ED.

\subsection{Error Analysis of SC-MAC}
If $x_{n}\in\{0,1\}$ and $y_{n}\in\mathcal{S}_{C}$ denote the OOK and $M-$PSK symbols of Alice and Charlie, respectively, then the received symbol at Bob, assuming symbol level synchronization, is given by
\bieee
r_{B,n} = \sqrt{1-\varepsilon}h_{AB,n}x_{n} + \sqrt{\varepsilon}h_{CB,n}y_{n} + w_{B,n},\hspace{0.5cm}1\leq n\leq L,\label{eq:rB_MAC}
\eieee
\noindent where $h_{AB,n}$, $h_{CB,n}$, and $w_{B,n}$ are as defined in previous sections. Unlike $3\phi$ DASC-MF, $r_{B,n}$ in SC-MAC are independent across time, therefore, we drop the subscript $n$ from the variables during the error analysis of SC-MAC. Owing to the non-coherent nature of Alice-to-Bob link of the MAC and coherent nature of the Charlie-to-Bob link of the MAC, and the Gaussian statistics of the channels and the noise, the distribution of $r_{B}$ conditioned on $x$ and $y$ as $r_{B}\vert_{x,y,h_{CB}}\sim\mathcal{CN}\left(\sqrt{\varepsilon}h_{CB} y, N_{o}\right)$, if $x=0$ and $r_{B}\vert_{x,y,h_{CB}}\sim\mathcal{CN}\left(\sqrt{\varepsilon}h_{CB} y, N_{o\varepsilon}\right)$, if $x=1$. Using the distribution of $r_{B}$, the joint MAP decoder for SC-MAC is
\bieee
\hat{i},\hat{m} = \arg\underset{\substack{i,m}}{\max\ } g_{MAC}\left(r_{B}\vert x = i, y = e^{\iota\frac{\pi}{M}\left(2m+1\right)}, h_{CB}\right),\label{eq:PDF_MAC}
\eieee
\noindent where $i\in\{0,1\}$ and $m\in\{0,\ldots,M-1\}$ and $g_{MAC}$ is the PDF of $r_{B}$ conditioned on $x$, $y$, and $h_{CB}$. If Alice and Charlie transmit OOK and $M-$PSK, then a transmitted pair is denoted by $(i,m)$. Fig.~\ref{fig:cons_MAC} depicts the constellation diagram at Bob jointly contributed by Alice and Charlie when using SC-MAC scheme with $M=4$. When $x_{n}=1$ and $x_{n}=0$, the effective variance of the noise at Bob is $N_{o}+1-\varepsilon$ and $N_{o}$, respectively. The \emph{square box} denotes the symbols with variance $N_{o}$ and \emph{circular disks} represent the symbols with variance $N_{o\varepsilon}=N_{o}+1-\varepsilon$. In the next section, we discuss the error analysis of the SC-MAC scheme.

\begin{table}[!htb]
    \caption{\label{tab:MAC_error}ERROR TERMS FOR SC-MAC AS GIVEN IN THEOREM~\ref{th:error_MAC}}
      \centering
      \scalebox{0.8}{
         \begin{tabular}[t]{|c c||c c||c c|} 
  \hline
$P_{1,MAC}=$ & $ 1 - e^{\frac{-\psi}{N_{o}}}$ & $P_{2,MAC}=$ & $Q_{1}\left(\frac{|C|}{\sqrt{N_{o}/2}}, \frac{\sqrt{\eta}}{\sqrt{N_{o}/2}}\right)$ & $P_{3,MAC}=$ & $Q\left(\frac{|h_{CB,n}| d}{\sqrt{2N_{o}}}\right)$
\\ \hline 
$P_{4,MAC}=$ & $ e^{\frac{-\psi}{N_{o\varepsilon}}}$ & $P_{5,MAC}=$ & $1 - Q_{1}\left(\frac{|D|}{\sqrt{N_{o\varepsilon}/2}}, \frac{\sqrt{\eta}}{\sqrt{N_{o\varepsilon}/2}}\right)$ & $P_{6,MAC}=$ & $Q\left(\frac{|h_{CB,n}| d}{\sqrt{2N_{o\varepsilon}}}\right)$
\\ \hline
   \end{tabular}
   } 
   \vspace{-0.25cm}
\end{table}

\begin{theorem}
\label{th:error_MAC}
At high SNR, a union bound on probability of error for SC-MAC is approximated as
\bieee
\frac{1}{2}\left(P_{1,MAC} + 2P_{2,MAC}+2P_{3,MAC} + P_{4,MAC} + 2P_{5,MAC}+ 2P_{6,MAC}\right), \label{eq:errorMAC3}
\eieee 
\noindent where the various error terms in \eqref{eq:errorMAC3} are tabulated in Table~\ref{tab:MAC_error}. Further, $\psi = \frac{N_{o}N_{o\varepsilon}}{N_{o}-N_{o\varepsilon}}\ln\left(\frac{N_{o}}{N_{o\varepsilon}}\right)$, $C = \frac{|h_{CB,n}| d N_{o}}{N_{o}-N_{o,\varepsilon}}$, $D = \frac{|h_{CB,n}| d N_{o\varepsilon}}{N_{o}-N_{o,\varepsilon}}$, and $\eta = \frac{N_{o}N_{o\varepsilon}}{N_{o}-N_{o\varepsilon}}\left[\ln\left(\frac{N_{o}}{N_{o\varepsilon}}\right) + \frac{|h_{CB,n}|^{2}d^{2}}{N_{o}-N_{o\varepsilon}}\right]$ are the parameters of the Marcum-Q function, where, $d = 2\sqrt{\varepsilon}\sin\frac{\pi}{M}$ denotes the minimum Euclidean distance between the constellation points received at Bob.
\end{theorem}
\begin{proof}
Let a pair corresponding to the symbols $x$ and $y$ be denoted by $(i,m)$. Let $\nabla_{(i,m)\rightarrow (i^{\prime},m^{\prime})}$ be the event $(i,m)\neq(i',m')$, then a transmitted pair $(i,m)$ is incorrectly decoded as $(i',m')$ if
\bieee
\nabla_{(i,m)\rightarrow (i^{\prime},m^{\prime})}\triangleq \frac{g_{MAC}\left(r_{B}\vert x = i, y = e^{\iota\frac{\pi}{M}\left(2m+1\right)}, h_{CB}\right)}{g_{MAC}\left(r_{B}\vert x = i^{\prime}, y = e^{\iota\frac{\pi}{M}\left(2m^{\prime}+1\right)}, h_{CB}\right)}\leq 1.\label{eq:event_MAC}
\eieee
\noindent Therefore, the probability of decoding a transmitted pair $(i,m)$ as $(i',m')$ is given as $\Pr((i,m)\rightarrow(i',m')) = \Pr\left(\nabla_{(i,m)\rightarrow (i^{\prime},m^{\prime})}\leq 1\right)$. If $\Pr\left((\hat{i},\hat{m})\neq (i,m)\right)$ denotes the overall probability of error in decoding a transmit pair $(i,m)$, then  $\Pr\left((\hat{i},\hat{m})\neq (i,m)\right)$ is upper bounded as 
\bieee
\Pr\left((\hat{i},\hat{m})\neq (i,m)\right) \leq \underset{(i,m)\neq (i^{\prime},m^{\prime})}{\sum_{i^{\prime}=0}^{1}\sum_{m^{\prime}=0}^{M-1}}\Pr\left(\nabla_{(i,m)\rightarrow (i^{\prime},m^{\prime})}\right),\label{eq:errorMAC1}
\eieee
\noindent where $(\hat{i},\hat{m})$ is the decoded pair at Bob, corresponding to the transmitted pair $(i,m)$. Thus, if $P_{e,MAC}$ denotes the overall probability of error for $2M$ symbols, then $P_{e,MAC}$ is upper bounded as $P_{e,MAC} \leq \frac{1}{2M}\sum_{i=0}^{1}\sum_{m=0}^{M-1} M\Pr((\hat{i},\hat{m})\neq (i,m))$. Finally, considering the nearest-neighbours at high SNR, $P_{e,MAC}$ is approximated as \eqref{eq:errorMAC3}, where the various error terms in \eqref{eq:errorMAC3} are tabulated in Table~\ref{tab:MAC_error}. Furthermore, these terms are also depicted in Fig.~\ref{fig:cons_MAC} for $4-$PSK used by Charlie.
\end{proof}

\begin{corollary}
\label{col:upper_MAC}
Using the results of Lemma~\ref{lm:P3_P1}, we upper bound $P_{2,MAC}$ and $P_{3,MAC}$ as $2P_{2,MAC}\leq P_{1,MAC}$ and $2P_{3,MAC}\leq P_{4,MAC}$. Further, using the Chernoff bound on the Q-functions and averaging $P_{e,MAC}$ over the realisations of $h_{CB}$, we get $\mathbb{E}_{h_{CB}}[P_{e,MAC}]\leq \mathcal{P}_{MAC} \triangleq P_{1,MAC} + P_{3,MAC}^{\star} + P_{4,MAC} + P_{6,MAC}^{\star}$, where $P_{3,MAC}^{\star} = \frac{2N_{o}}{4N_{o} + d^{2}}$ and $P_{6,MAC}^{\star} = \frac{2N_{o\varepsilon}}{4N_{o\varepsilon} + d^{2}}$. 
\end{corollary}

\begin{figure}[t]
\vspace{-0.5cm}
\captionsetup{width=0.48\textwidth}
\begin{center}
    \begin{minipage}[t]{0.48\textwidth}
    \centering
\includegraphics[scale=0.45]{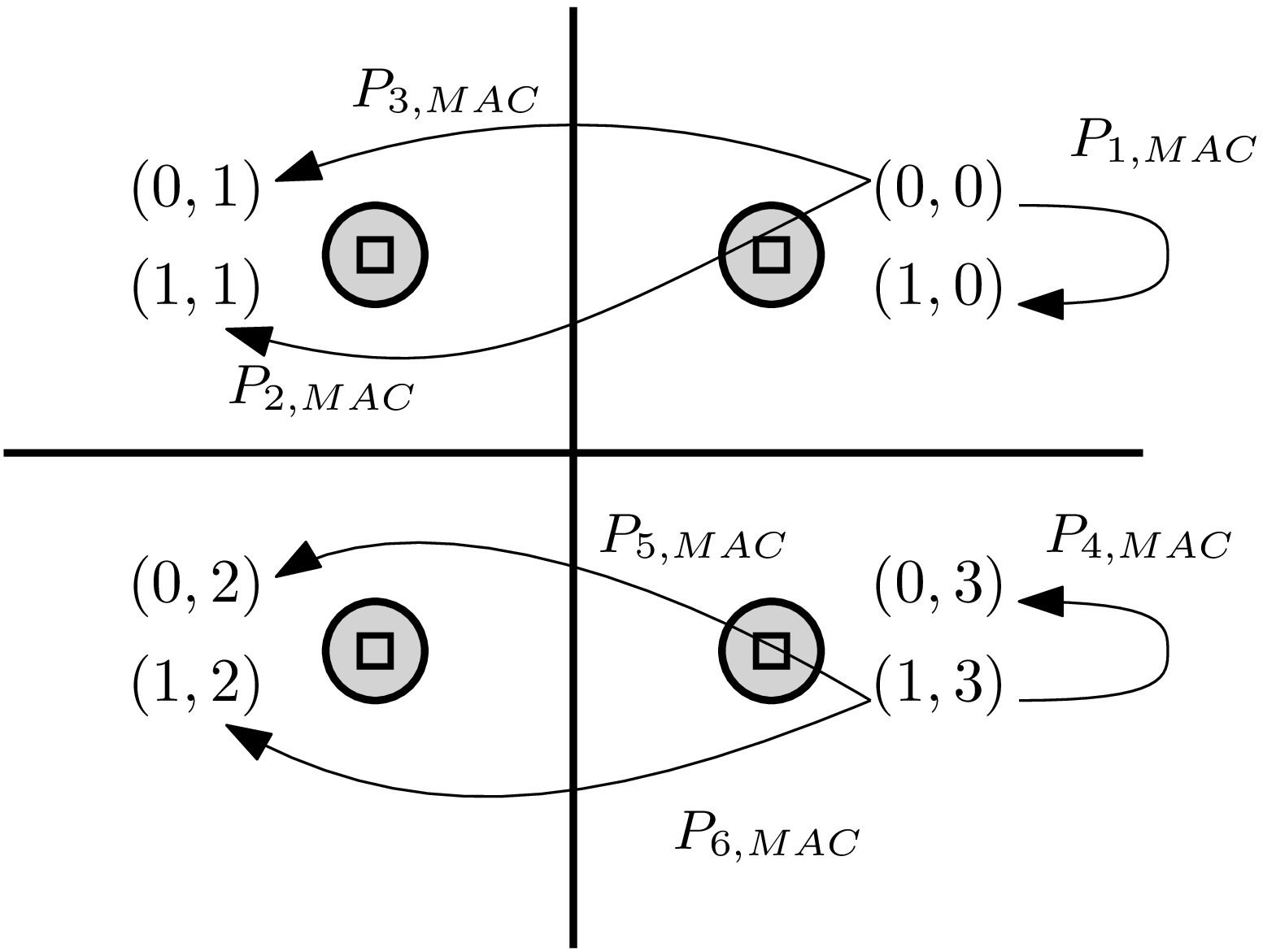}
\caption{\label{fig:cons_MAC} Constellation diagram jointly contributed by Alice and Charlie when using SC-MAC.}
    \end{minipage}%
    \hfill
   \begin{minipage}[t]{0.48\textwidth}
\centering
        \includegraphics[scale = 0.45]{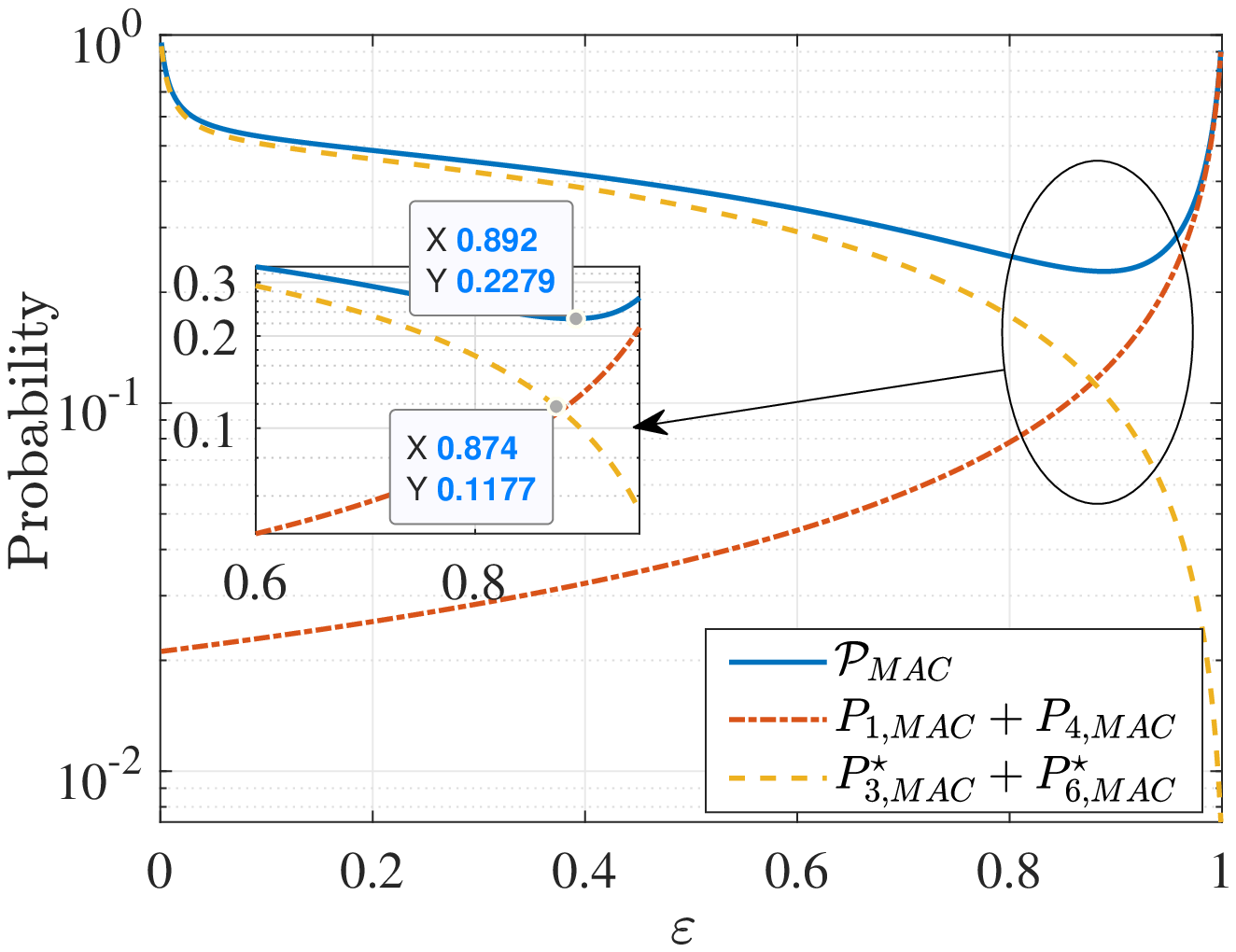}
\caption{\label{fig:intersection_MAC} Variation of $\mathcal{P}_{MAC}$ and its increasing and decreasing terms as a function of $\varepsilon$ at $25$ dB.}
    \end{minipage}%
    \end{center}
    \vspace{-0.5cm}
\end{figure}

While observing \eqref{eq:rB_MAC}, we notice that when $\varepsilon\rightarrow 1$, Alice's symbols are transmitted with lower energy, whereas Charlie transmits his symbols with higher energy. However, when $\varepsilon\rightarrow 0$, Alice's symbols are transmitted with  higher energy as compared to Charlie. Thus, in the extreme range, the joint error-rates at Bob are expected to be high and therefore,  Alice and Charlie must use a value of $\varepsilon$ that minimizes the error-rates at Bob. Towards minimizing $\mathcal{P}_{MAC}$ at Bob, we pose an optimization problem as given below.
\begin{mdframed}
\bieee
\varepsilon^{\dagger} = \arg\underset{\varepsilon}{\min} \quad &  & \mathcal{P}_{MAC}; \ \ \text{s.t.: }  0<\varepsilon < 1.\label{optMAC}
\eieee
\end{mdframed}

Towards minimizing $\mathcal{P}_{MAC}$, we note that, it is straightforward to prove that $\mathcal{P}_{MAC}$ is a sum of increasing and decreasing terms w.r.t. $\varepsilon$. In particular, using the results of Lemma~\ref{lm:inc_dec_P3}, we notice that $P_{1,MAC}+P_{4,MAC}$ is an increasing function of $\varepsilon$ and $P_{3,MAC}^{\star}+P_{6,MAC}^{\star}$ is a decreasing function of $\varepsilon$ and the intersection of the increasing and decreasing terms is close to the unique minima of $\mathcal{P}_{MAC}$. This observation is exemplified in Fig.~\ref{fig:intersection_MAC} at SNR = 25 dB. Along the similar lines of Theorem~\ref{th:p3_dec}, it is straightforward to prove that, $P_{1,MAC}+P_{4,MAC}$ and $P_{3,MAC}^{\star}+P_{6,MAC}^{\star}$ intersect at a unique value of $\varepsilon\in(0,1)$. Therefore, the unique intersection of $P_{1,MAC}+P_{4,MAC}$ and $P_{3,MAC}^{\star}+P_{6,MAC}^{\star}$ can be computed using the NR algorithm. In the next section, we present Monte-Carlo simulations for SC-MAC scheme and showcase its error performance.

\section{Simulation Results on the $3\phi$ DASC-MF and SC-MAC Schemes}
\label{sec:sim}

\blue{In this section, we present simulation results for the two schemes proposed in the previous sections. The simulation parameters for the Monte-Carlo simulations are as follows: We assume $4$-PSK signalling at Charlie and OOK signalling at Alice. All the links in the network are Rayleigh faded, such that $\sigma_{AB}^{2}=\sigma_{CB}^{2}=1$ and $\sigma_{AC}^{2} = 4$, thus providing $6$ dB improvement in SNR on Alice-to-Charlie link as compared to  Alice-to-Bob link of the MAC. Further, we assume Bob has the perfect channel estimates of Charlie-to-Bob link on $f_{CB}$. Furthermore, all the channels and noise realizations are statistically independent. Finally, the variance of the AWGN at Bob and Charlie is $N_{o} = SNR^{-1}$. For the $3\phi$ scheme, we use $L=200$ and $\Delta = 0.1$ along with the following three delay models, i.e.,  $\Theta = \lceil 10\log N_{C}^{\dagger}\rceil$, $\Theta = \lceil\log_{2} N_{C}^{\dagger}\rceil$, and $\Theta=\lceil\frac{L}{2}\rceil$. Since, $\Theta$ is a function of the FD architecture of Charlie, similar results can be generated for any model on $\Theta$, as long as, $\Theta\leq\frac{L}{2}$.}

\begin{table}[!h]
%\vspace{-0.25cm}
\begin{center}
\caption{\label{tab:NC_beta_tab} VALUES OF $(N_{C}^{\dagger},\beta^{\dagger})$ USING EXACT AND INTERSECTION METHODS FOR VARIOUS $\Theta$}
\scalebox{0.85}{
\begin{tabular}{|c|ll|ll|ll|}
\hline
\multicolumn{1}{|l|}{\multirow{2}{*}{$\downarrow \Theta$ $\setminus$ $\underset{\rightarrow}{\text{SNR}}$}} & \multicolumn{2}{c|}{21 dB}    & \multicolumn{2}{c|}{25 dB}    & \multicolumn{2}{c|}{29 dB}    
\\ \cline{2-7} 
\multicolumn{1}{|l|}{} & \multicolumn{1}{l|}{$(N_{C}^{\dagger}, \beta^{\dagger})_{\textrm{EM}}$} & $(N_{C}^{\dagger}, \beta^{\dagger})_{\textrm{Int}}$ & \multicolumn{1}{l|}{$(N_{C}^{\dagger}, \beta^{\dagger})_{\textrm{EM}}$} & $(N_{C}^{\dagger}, \beta^{\dagger})_{\textrm{Int}}$  & \multicolumn{1}{l|}{$(N_{C}^{\dagger}, \beta^{\dagger})_{\textrm{EM}}$} & $(N_{C}^{\dagger}, \beta^{\dagger})_{\textrm{Int}}$ 
\\ \hline
$\lceil 10\log_{10}N_{C}\rceil$                                 & \multicolumn{1}{l|}{(125,0.8730)} & (125, 0.9020) & \multicolumn{1}{l|}{(127,0.9220)} & (127,0.9310) & \multicolumn{1}{l|}{(129, 0.9580)} & (129, 0.9570)  
\\ \hline
$\lceil\log_{2}N_{C}\rceil$                                     & \multicolumn{1}{l|}{(126,0.8730)} & (126, 0.9020) & \multicolumn{1}{l|}{(128,0.9220)} & (128,0.9310)  & \multicolumn{1}{l|}{(130, 0.9580)} &  (130, 0.9570)
   \\ \hline
$\lceil\sqrt{N_{C}}\rceil$                                & \multicolumn{1}{l|}{(122,0.8730)} & (122, 0.9020) & \multicolumn{1}{l|}{(127,0.9220)} & (127,0.9310) & \multicolumn{1}{l|}{(129, 0.9580)} & (129, 0.9570)
\\ \hline
\end{tabular}
}
\end{center}
\vspace{-0.5cm}
\end{table}

\begin{figure}
\captionsetup{width=0.48\textwidth}
\begin{center}
\begin{minipage}[t]{0.48\textwidth}
         \captionsetup{type=table,font=scriptsize}
    \caption{\label{tab:ES} \blue{VALUES OF $(N_{C}^{*},\beta^{*})$ OBTAINED USING EXHAUSTIVE METHOD FOR VARIOUS $\Theta$}}
    \scalebox{0.8}{
\begin{tabular}[t]{|c|c|c|c|}
\hline
 $\downarrow \Theta\setminus \underset{\rightarrow}{SNR}$& 21 dB & 25 dB & 29 dB\\ \hline
$\lceil 10\log_{10}N_{C}\rceil$ & (125, 0.8700) & (127, 0.9290) & (129, 0.9430)
\\ \hline
$\lceil \log_{2}N_{C}\rceil$ & (126, 0.8700) & (128, 0.9290) & (130, 0.9430)
\\ \hline
$\lceil \sqrt{N_{C}}\rceil$ & (122, 0.8700) & (127, 0.9290) & (129, 0.9430) 
\\ \hline
\end{tabular}
}
    \end{minipage}%
    \hfill  
\begin{minipage}[t]{0.48\textwidth}
   \captionsetup{type=table,font=scriptsize}
    \caption{\label{tab:MAC} VALUES OF $\varepsilon$ OBTAINED USING EXHAUSTIVE AND INTERSECTION METHODS FOR SC-MAC}
    \scalebox{0.9}{
\begin{tabular}[t]{|c|p{2.5em}|p{2.5em}|p{2.5em}|p{2.5em}|}
\hline
 $\downarrow \varepsilon\setminus \underset{\rightarrow}{SNR}$& $19$ dB & 21 dB & 25 dB & 29 dB \\ \hline
$\varepsilon_{\textrm{ES}}^{*}$ & 0.8080 & 0.8700 & 0.8920 & 0.9250 \\ \hline
$\varepsilon_{\textrm{Int}}^{\dagger}$ & 0.7820 & 0.8510 & 0.8740 & 0.9190 \\ \hline
\end{tabular}
}
    \end{minipage}% 
    \end{center}
\end{figure}

\blue{To present the error performance of the $3\phi$ scheme, we first compute the near-optimal values of $(N_{C},\beta)$ using the following three methods: i) an Exhaustive Search (ES), ii) an Exact Method (EM), and iii) the $N_{C}$-$\beta$ Optimization algorithm. For the ES, we use Monte-Carlo simulations to empirically compute the overall probability of error in \eqref{eq:Pe_main} and then exhaustively search for $(N_{C}^{*},\beta^{*})$ that minimises \eqref{eq:Pe_main}, for a given $\Theta$ and SNR. The step size for $N_{C}$ is one and that for $\beta\in(0,1)$ is $10^{-3}$. For the EM, we use the upper bounds proposed in Corollary~\ref{col:upper_Phase2} and Corollary~\ref{col:upper_Phase3} and exhaustively search for $(N_{C}^{\dagger},\beta^{\dagger})$ that minimises \eqref{eq:Pe_main2}, for a given $\Theta$ and SNR. Here, the resolution of $\beta$ is same as used for the ES. Finally, we compute $(N_{C}^{\dagger},\beta^{\dagger})$ using the $N_{C}$-$\beta$ Optimization algorithm, as explained in Algorithm~\ref{Algo}. In Table~\ref{tab:NC_beta_tab}, we tabulate $N_{C}^{\dagger}$ and $\beta^{\dagger}$ for various combinations of $\Theta$ and SNR using the EM and $N_{C}$-$\beta$ Optimization algorithm. Further, for the same $\Theta$ and SNR, in Table~\ref{tab:ES}, we also tabulate $N_{C}^{\dagger}$ and $\beta^{\dagger}$ obtained using the ES. From Table~\ref{tab:NC_beta_tab} and Table~\ref{tab:ES} we observe that the obtained values of $N_{C}$ and $\beta$ are close. Further, we also observe that the required value of $N_{C}^{\dagger}$ increases with SNR. This is due to the fact that, although we have upper bounded the interference from Alice on Alice-to-Bob link of the MAC by $\Delta N_{o}$, Charlie uses the same energy, $\Delta N_{o}$ to decode Alice's symbols. Therefore, as SNR increases, $N_{o}$ decreases, and Charlie requires more receive-antennas to faithfully decode Alice's symbols. Along the similar lines, for the SC-MAC scheme, we tabulate $\varepsilon$ obtained using the ES method and the intersection method, for various SNRs in Table~\ref{tab:MAC}. For ES we exhaustively search for the value of $\varepsilon\in(0,1)$ (in the steps of $10^{-3}$) that minimises the empirically computed average probability of error. From the table we infer that $\varepsilon$ obtained using both the methods are close.}

\begin{figure}
\vspace{-0.5cm}
\captionsetup{width=0.49\textwidth}
\begin{center}
    \begin{minipage}[t]{0.48\textwidth}
        \centering
        \includegraphics[width = 0.8\textwidth, height = 0.7\textwidth]{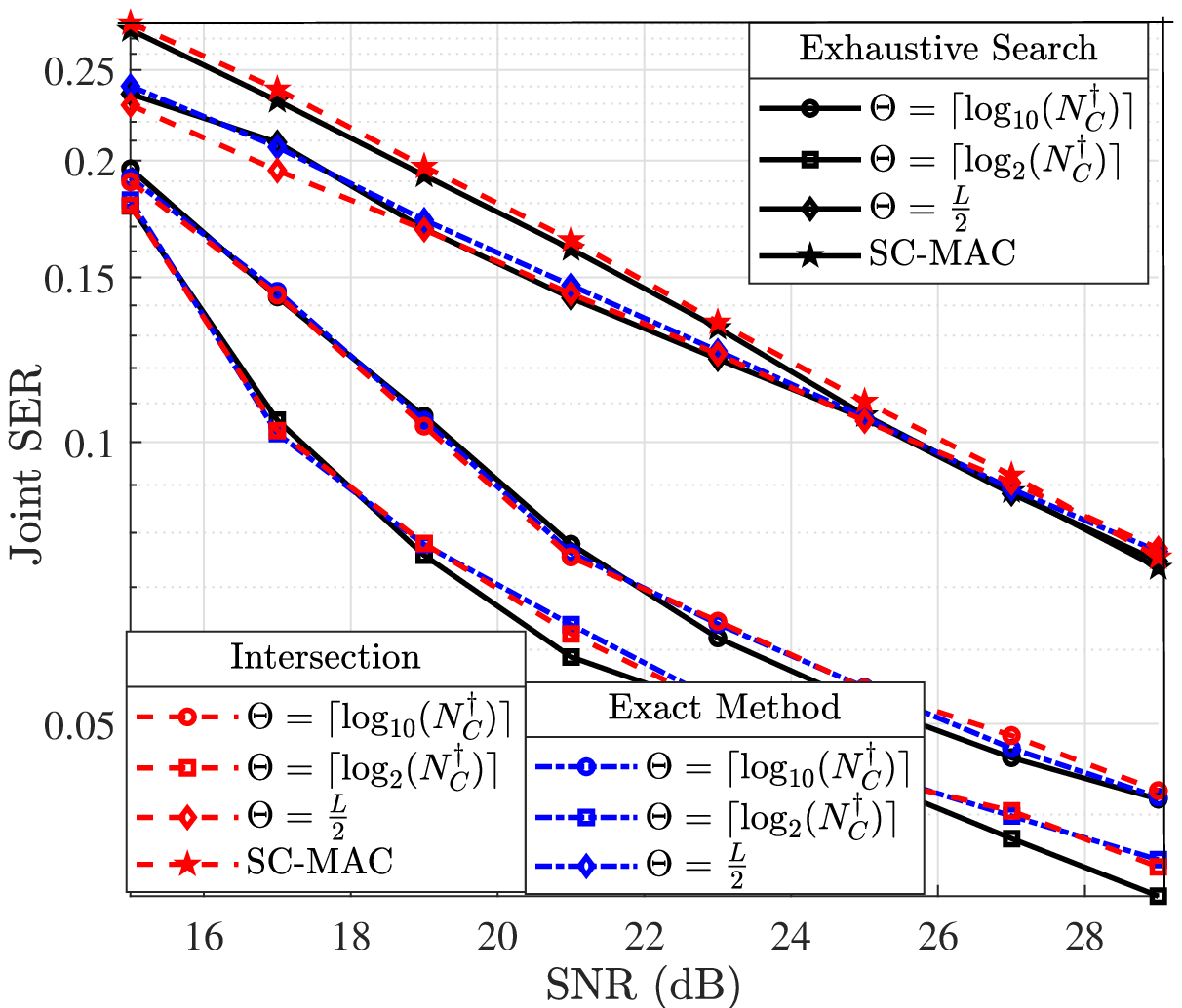}
\caption{\label{fig:joint_error}\blue{Joint SER when using $3\phi$ DASC-MF and SC-MAC.}}
    \end{minipage}%
    \hfill
   \begin{minipage}[t]{0.48\textwidth}
        \centering
        \includegraphics[width = 0.8\textwidth, height = 0.7\textwidth]{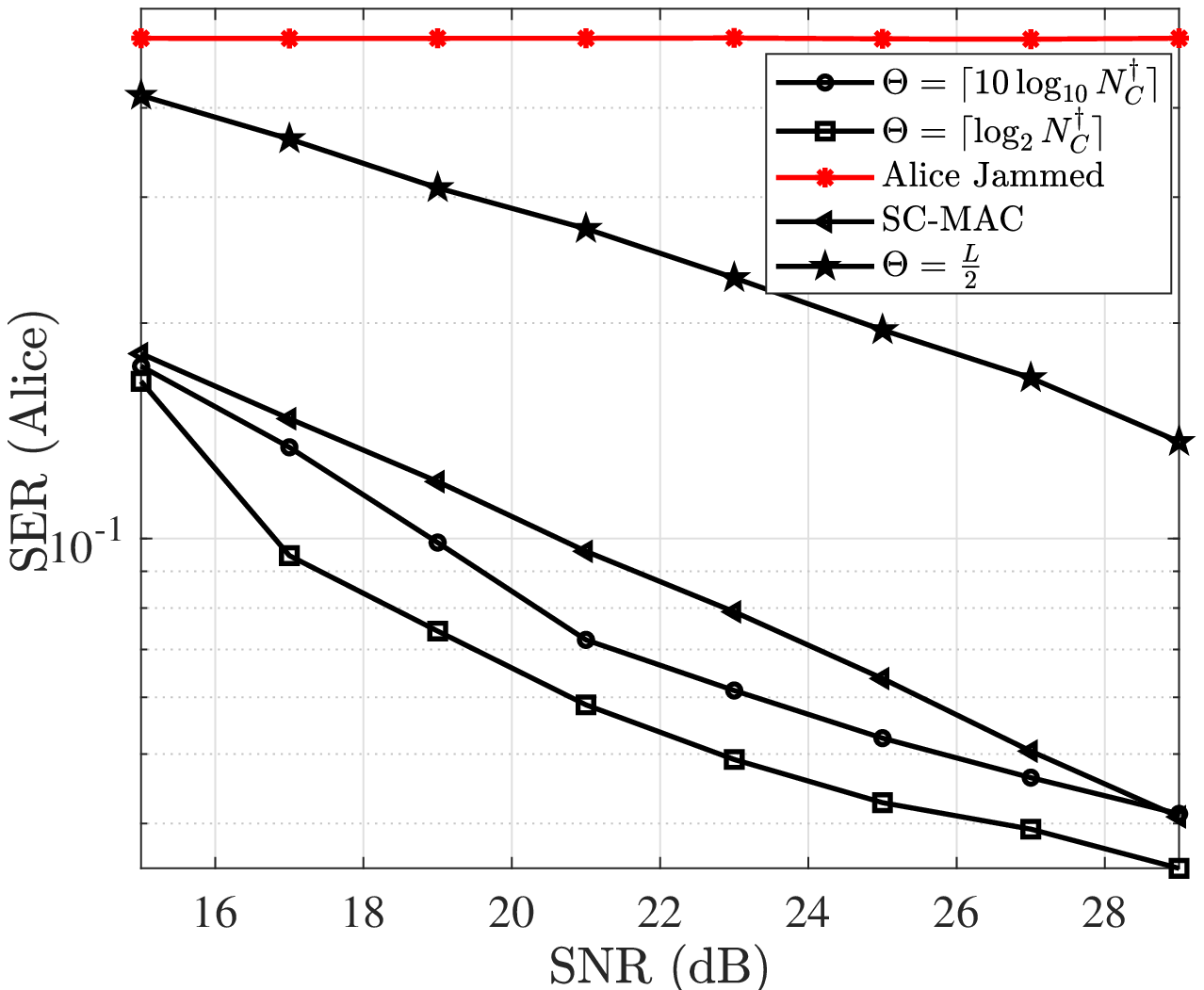}
\caption{\label{fig:alice_error}\blue {SER (Alice) using $3\phi$ DASC-MF and SC-MAC.}}
    \end{minipage}%
    \end{center}
    \vspace{-0.5cm}
\end{figure}

\blue{In Fig.~\ref{fig:joint_error}, we plot the Joint Symbol Error Rate (SER), i.e., the joint error performance of Alice and Charlie, at Bob when using the $3\phi$ DASC-MF scheme (for various models on $\Theta$) and SC-MAC scheme, as a function of SNR. In $3\phi$ DASC-MF, for each model, we first obtain the value of $(N_{C},\beta)$ pair using the ES, EM, and $N_{C}$-$\beta$ Optimization Algorithm and then obtain the value of Joint SER through Monte-Carlo simulations. Along the similar lines, for SC-MAC scheme, we first obtain the value of $\varepsilon$ using the ES method and the intersection method and then compute the Joint SER. It is evident that the Joint SER for both the schemes reduces as a function of SNR. We also note that among the three models considered for $3\phi$ DASC-MF scheme, the error performance is best when $\Theta = \lceil\log_{2}N_{C}^{\dagger}\rceil$ as $\frac{\Theta}{L}$ is minimum when $\Theta = \lceil\log_{2}N_{C}^{\dagger}\rceil$, thereby reducing the fraction of symbols decoded during Phase-III. Additionally, when comparing SC-MAC and $3\phi$ DASC-MF scheme, we note that SC-MAC is sub-optimal for the assumed models on $\Theta$. This is because, throughout $3\phi$ DASC-MF scheme, Bob decodes a majority of Alice's symbols using the coherent Charlie-to-Bob link of the MAC, whereas, in SC-MAC, Bob decodes all  Alice's symbols using the non-coherent Alice-to-Bob link of the MAC. Thus, Alice's symbols in SC-MAC are decoded more reliably than $3\phi$ DASC-MF.}\looseness = -1

 \blue{Finally, in Fig.~\ref{fig:alice_error}, we plot the improvement in Alice's error performance before and after implementing the proposed schemes. It is interesting to note that although the joint error performance of SC-MAC is inferior to $3\phi$ DASC-MF when $\Theta=\frac{L}{2}$, SC-MAC outperforms $3\phi$ DASC-MF when $\Theta = \frac{L}{2}$ when considering only Alice's SER at Bob. This is because, in SC-MAC, Bob decodes two symbols i.e., Alice's current symbols and Charlie's current symbol during the entire duration. In contrast, when $\Theta=\frac{L}{2}$, Bob decodes three symbols, i.e, Alice current symbol, Charlie's current symbol and Alice's symbol delayed by $\Theta$ symbols. As a result, the error rates for Alice are high for $\Theta = \frac{L}{2}$. Overall, it is clear from the plots that SC-MAC helps Alice to evade the jamming attack from Dave within the latency constraint, however, it is sub-optimal as compared to $3\phi$ DASC-MF scheme. Nevertheless, it is a still a good option when $\Theta\geq\frac{L}{2}$. }

\section{Covertness Analysis}
\label{sec:Covertness}
 As discussed in Sec.~\ref{sec:system}, Dave uses an ED to detect the changes in the average energy level on all the frequencies in the network after the jamming attack. Since Alice and Charlie use $f_{AB}$ and $f_{CB}$ to execute the countermeasures, we restrict our analysis to these frequencies only. Thus, in the context of this work, we refer to a mitigation scheme as covert if Dave is unable to detect changes in the average energy levels of $f_{AB}$ and $f_{CB}$. Further, throughout this section, we assume that Dave has the knowledge of the channel statistics of Alice-to-Dave link and Charlie-to-Dave link, but no knowledge of the channel coefficients of these links. 

\subsection{Covertness Analysis for ED when using $3\phi$ DASC-MF Relaying Scheme}
\label{ssec:covert_ED3p}

In this section, we first discuss energy detection at Dave on $f_{AB}$ and $f_{CB}$ when using $3\phi$ DASC-MF and then, compute probability of false-alarms and probability of miss-detections on $f_{AB}$ and $f_{CB}$. Before executing the jamming attack, Dave collects a frame of $L$ symbols on $f_{AB}$ and $f_{CB}$ and computes their average energy. The $n^{th}$ symbol received at Dave, before the jamming attack on $f_{AB}$ and $f_{CB}$ are respectively given as
\bieee
r_{D,AB,n} &=& h_{AD,n}x_{n} + w_{D,AB,n},\hspace{0.5cm} 1\leq n\leq L,\label{eq:rd1}\\
r_{D,CB,n} &=& h_{CD,n}y_{n} + w_{D,CB,n},\hspace{0.5cm} 1\leq n\leq L,\label{eq:rd2}
\eieee 
\noindent where $h_{AD,n}\sim\mathcal{CN}(0,1)$ and $h_{CD,n}\sim\mathcal{CN}(0,1+\partial)$ are the $n^{th}$ channel coefficients of Alice-to-Dave link and Charlie-to-Dave link, respectively, such that, $\partial$ captures the relative difference in the variance of Alice-to-Dave link and Charlie-to-Dave link. Further, $w_{D,AB,n}\sim\mathcal{CN}(0,\tilde{N}_{o})$ is the effective AWGN at Dave on $f_{AB}$, such that, $\tilde{N}_{o} = N_{o}+\sigma_{DD}^{2}$, where $\sigma_{DD}^{2}$ is the variance of the residual SI at Dave and $N_{o}$ is the variance of the AWGN at Dave. Furthermore, $w_{D,CB,n}\sim\mathcal{CN}(0,N_{o})$ is the AWGN at Dave on $f_{CB}$. Since $r_{D,AB,n}$ and $r_{D,CB,n}$ are statistically independent over $n$, using weak law of large numbers, $\frac{1}{L}\sum_{n=1}^{L}\vert r_{D,AB,n}\vert^{2}\rightarrow \mathcal{E}_{AB}$ and $\frac{1}{L}\sum_{n=1}^{L}\vert r_{D,CB,n}\vert^{2}\rightarrow \mathcal{E}_{CB}$, in probability, such that $\mathcal{E}_{AB} = \tilde{N}_{o} + 0.5$ and $\mathcal{E}_{CB} = N_{o}+1+\partial$. However, due to short packet length, the measured average energies on $f_{AB}$ and $f_{CB}$ are not equal to $\mathcal{E}_{AB}$ and $\mathcal{E}_{CB}$. Therefore, if $\mathcal{H}_{0}$ denotes the hypothesis that no countermeasure is implemented, then, probability of false-alarms on $f_{AB}$ and $f_{CB}$ are defined as follows.
\begin{definition}
\label{def:pfa_ab}
The probability of false-alarm on $f_{AB}$, denoted by, $\mathbf{P}_{FA, AB}$ is given as $\mathbf{P}_{FA, AB} = \Pr\left(\vert\mathcal{U}_{L,AB} - \mathcal{E}_{AB}\vert\geq \nu_{AB}\vert \mathcal{H}_{0}\ true\right)$, where, $\mathcal{U}_{L,AB}$ denotes the random variable (RV) corresponding to the average energy of $L$ symbols received at Dave on $f_{AB}$ and $\nu_{AB}>0$ is a parameter of Dave's choice.
\end{definition}
\begin{definition}
\label{def:pfa_cb}
The probability of false-alarm on $f_{CB}$, denoted by, $\mathbf{P}_{FA, CB}$ is given as $\mathbf{P}_{FA, CB} = \Pr\left(\vert\mathcal{U}_{L,CB} - \mathcal{E}_{CB}\vert\geq \nu_{CB}\vert \mathcal{H}_{0}\ true\right)$, for $\nu_{CB}>0$, where, $\mathcal{U}_{L,CB}$ denotes the RV corresponding to the average energy of $L$ symbols received at Dave on $f_{CB}$ and $\nu_{CB}$ is a parameter of Dave's choice.
\end{definition}

When $\tilde{N}_{o}\ll 1$, the distribution of $\mathcal{U}_{L,AB}$ is approximated as $\frac{1}{2^{L}}\sum_{l=0}^{L}{L \choose l}\mathrm{G}\left(l,\frac{1}{L}\right)$, where $\mathrm{G}(\cdot,\cdot)$ denotes the Gamma distribution \cite[Theorem 5]{my_TCOM}. Along the similar lines, when $N_{o}\ll 1$, the distribution of $\mathcal{U}_{L,CB}$ can be approximated as $\frac{1}{2^{L}}\sum_{l=0}^{L}{L \choose l}\mathrm{G}\left(l,\frac{1+\partial}{L}\right)$. Thus, the expressions of $\mathbf{P}_{FA, AB}$ and $\mathbf{P}_{FA, CB}$ when $\tilde{N}_{o}\ll 1$ and $N_{o}\ll 1$, respectively, are approximated as
\bieee
\mathbf{P}_{FA, AB} &\approx & \frac{1}{2^{L}}\left[\sum_{l=0}^{L}{L \choose l}\frac{\gamma\left(l,L\left(\mathcal{E}_{AB} - \nu_{AB}\right)\right)}{\Gamma(l)} + \sum_{l=0}^{L}{L \choose l}\frac{\Gamma\left(l,L\left(\mathcal{E}_{AB} + \nu_{AB}\right)\right)}{\Gamma(l)}\right],\label{eq:pfa_ab}\\
\mathbf{P}_{FA, CB} &\approx & \frac{1}{2^{L}}\left[\sum_{l=0}^{L}{L \choose l}\frac{\gamma\left(l,\frac{L}{1+\partial}\left(\mathcal{E}_{CB} - \nu_{CB}\right)\right)}{\Gamma(l)} + \sum_{l=0}^{L}{L \choose l}\frac{\Gamma\left(l,\frac{L}{1+\partial}\left(\mathcal{E}_{CB} + \nu_{CB}\right)\right)}{\Gamma(l)}\right],\label{eq:pfa_Cb}
\eieee
\noindent where $\nu_{AB}$ and $\nu_{CB}$ are decided by Dave. However, small values of $\nu_{AB}$ and $\nu_{CB}$ often result in high false-alarms, and large values of $\nu_{AB}$ and $\nu_{CB}$ result in high miss-detection. Thus, the values of $\nu_{AB}$ and $\nu_{CB}$ must be cautiously chosen by Dave so as to minimize both probability of false-alarm and probability of miss-detection. In the rest of the section, we discuss the probability of miss-detection on $f_{AB}$ and $f_{CB}$, when using $3\phi$ DASC-MF scheme for a given choice of $\nu_{AB}$ and $\nu_{CB}$. Towards computing the probability of miss-detection, we first present the symbols observed at Dave on $f_{AB}$ and then characterise the probability of miss-detection on $f_{AB}$. Along the similar lines, we will then characterise the probability of miss-detection on $f_{CB}$. We recall that Alice and Charlie use OOK symbols from a pre-shared Gold-sequence to cooperatively pour their residual energies on $f_{AB}$. Therefore, if $b_{n}\in\{0,1\}$ denotes the $n^{th}$ OOK bit jointly transmitted by Alice and Charlie, and $r_{D,AB,n}^{\dagger}$ denotes the $n^{th}$ baseband symbol received at Dave, then
\begin{subnumcases}{r_{D,AB,n}^{\dagger}=}
\sqrt{\alpha}h_{AD,n}b_{n} + \sqrt{1-\alpha}h_{CD,n}b_{n} + w_{D,AB,n}, & if $1\leq n\leq L-\Theta$,\label{eq:rd3}\\
\sqrt{\beta^{\dagger}}h_{AD,n}b_{n} + \sqrt{1-\beta^{\dagger}}h_{CD,n}b_{n} + w_{D,AB,n}, & if $L-\Theta+1\leq n\leq L$.\label{eq:rd4}
\end{subnumcases}

Let $\mathcal{H}_{1}$ denote the hypothesis that a countermeasure is implemented, then, formally, the probability of miss-detection at Dave on $f_{AB}$ is given as follows.
\begin{definition}
\label{def:pd_ab}
Let $\mathcal{V}_{L,AB}$ denote the RV corresponding to the average energy of $L$ symbols on $f_{AB}$ after using $3\phi$ DASC-MF scheme. Thus, probability of miss-detection denoted by, $\mathbf{P}_{MD,AB}^{3\phi}$, is $\mathbf{P}_{MD,AB}^{3\phi} = \Pr\left(\vert\mathcal{V}_{L,AB} - \mathcal{E}_{AB}\vert \leq \nu_{AB}\vert \mathcal{H}_{1}\text{ true}\right)$, for $\nu_{AB}>0$.
\end{definition}
\noindent Using \eqref{eq:rd3} and \eqref{eq:rd4}, we will now characterize $\mathbf{P}_{MD,AB}^{3\phi}$. Let $v_{n_{1},AB}$ denote the RV corresponding to $\vert r_{D,AB,n_{1}}^{\dagger}\vert^{2}$, such that $1\leq n_{1}\leq L-\Theta$ and $v_{n_{2},AB}$ denote the RV corresponding to $\vert r_{D,AB,n_{2}}^{\dagger}\vert^{2}$, such that $L-\Theta+1\leq n_{2}\leq L$. Thus, $\mathcal{V}_{L,AB} = \mathcal{V}_{L_{1},AB} + \mathcal{V}_{L_{2},AB}$, such that $\mathcal{V}_{L_{1},AB}  = \frac{1}{L}\sum_{n_{1}=1}^{L-\Theta}v_{n_{1},AB}$ and $\mathcal{V}_{L_{2},AB} = \frac{1}{L}\sum_{n_{2}=L-\Theta+1}^{L}v_{n_{2},AB}$. In the next lemma, we compute the PDF of $\mathcal{V}_{L_{1},AB}$ and $\mathcal{V}_{L_{2},AB}$.
\begin{lemma}
\label{lm:pfd_md1}
When $\tilde{N}_{o}\ll 1$, the distribution of $\mathcal{V}_{L_{1},AB}$ is approximated as $\frac{1}{2^{L-\Theta}}\sum_{l=0}^{L-\Theta}{L-\Theta \choose l}\mathrm{G}\left(l,\frac{\mathcal{A}}{L}\right)$, where $\mathcal{A} = \alpha + (1-\alpha)(1+\partial)$. Similarly, the distribution of $\mathcal{V}_{L_{2},AB}$ is approximated as $\frac{1}{2^{\Theta}}\sum_{l=0}^{\Theta}{\Theta \choose l}\mathrm{G}\left(l,\frac{\mathcal{B}}{L}\right)$, where $\mathcal{B} = \beta^{\dagger} + (1-\beta^{\dagger})(1+\partial)$.
\end{lemma}
\noindent Finally, $\frac{1}{2^{L-\Theta}}\sum_{l=0}^{L-\Theta}{L-\Theta \choose l}\mathrm{G}\left(l,\frac{\mathcal{A}}{L}\right)\ast\frac{1}{2^{\Theta}}\sum_{l=0}^{\Theta}{\Theta \choose l}\mathrm{G}\left(l,\frac{\mathcal{B}}{L}\right)$ gives the distribution of $\mathcal{V}_{L,AB}$, where $\ast$ denotes linear convolution. Using the distribution of $\mathcal{V}_{L,AB}$, it is straightforward to compute $\mathbf{P}_{MD,AB}^{3\phi}$. 
\begin{rem}
\label{rem:sum_pfa_pmd_ED}
If $\partial = 0$, then $\mathbf{P}_{FA,AB}+\mathbf{P}_{MD,AB}^{3\phi}=1$ for all $\nu_{AB}>0$.
\end{rem}
\noindent Although, Remark~\ref{rem:sum_pfa_pmd_ED} theoretically guarantees that $\mathbf{P}_{FA,AB}+\mathbf{P}_{MD,AB}^{3\phi}=1$, in Sec.~\ref{ssec:simulation_covert}, through Monte-Carlo simulations, we show that, $\mathbf{P}_{FA,AB}+\mathbf{P}_{MD,AB}^{3\phi}$ is close to $1$ for various values of $\partial$.

We now compute the probability of miss-detection on $f_{CB}$, denoted by, $\mathbf{P}_{MD,CB}^{3\phi}$. Towards computing $\mathbf{P}_{MD,CB}^{3\phi}$, we first observe the symbol received at Dave on $f_{CB}$, denoted by $r_{D,CB,n}^{\dagger}$ and then characterise its PDF. Thus, the $n^{th}$ baseband symbol received at Dave on $f_{CB}$ is given as
\begin{subnumcases}{r_{D,CB,n}^{\dagger}=}
 \sqrt{1-\alpha}h_{AD,n}x_{n} + \sqrt{\alpha}h_{CB,n}y_{n} + w_{CB,n}, & if $1\leq n\leq \Theta$,\label{eq:dave_fcb1}
\\
\sqrt{1-\alpha}h_{AD,n}x_{n} + h_{CB,n}t_{n} + w_{CB,n}, & if $\Theta+1\leq n\leq L-\Theta$,\label{eq:dave_fcb2}\\
\sqrt{1-\beta^{\dagger}}h_{AD,n}x_{n} + \sqrt{\beta^{\dagger}}h_{CB,n}s_{n} + w_{CB,n}, & if $L-\Theta+1\leq n\leq L$,\label{eq:dave_fcb3}
\end{subnumcases}
\noindent where $\alpha = 1-\Delta N_{o}$. Let $v_{n_{1},CB}$, $v_{n_{2},CB}$, and $v_{n_{3},CB}$ denote the RVs corresponding to $\vert r_{D,CB,n_{1}}^{\dagger}\vert^{2}$, $\vert r_{D,CB,n_{2}}^{\dagger}\vert^{2}$, and $\vert r_{D,CB,n_{3}}^{\dagger}\vert^{2}$, respectively, such that $1\leq n_{1}\leq \Theta$, $\Theta+1\leq n_{2}\leq L-\Theta$, and $L-\Theta+1\leq n_{3}\leq L$. Thus, the RV corresponding to the average energy of $L$ symbols on $f_{CB}$, denoted by $\mathcal{V}_{L,CB}$ is given as $\mathcal{V}_{L,CB} = \mathcal{V}_{L_{1},CB} + \mathcal{V}_{L_{2},CB}+\mathcal{V}_{L_{3},CB}$, such that $\mathcal{V}_{L_{1},CB} = \frac{1}{L}\sum_{n_{1} = 1}^{\Theta}v_{n_{1},CB}$, $\mathcal{V}_{L_{2},CB} = \frac{1}{L}\sum_{n_{2} = \Theta+1}^{L-\Theta}v_{n_{2},CB}$, and $\mathcal{V}_{L_{3},CB} = \frac{1}{L}\sum_{n_{3} = L-\Theta+1}^{L}v_{n_{3},CB}$. Using $\mathcal{V}_{L,CB}$, we formally define probability of miss-detection at Dave when measuring the average energy level on $f_{CB}$.
\begin{definition}
\label{def:pd_cb}
Given $\mathcal{H}_{1}$ is true, probability of miss-detection when Dave measures the average energy level on $f_{CB}$, denoted by $\mathbf{P}_{MD,CB}^{3\phi}$, is given as $\mathbf{P}_{MD,CB}^{3\phi} = \Pr\left(\vert\mathcal{V}_{L,CB}-\mathcal{E}_{CB}\vert\leq\nu_{CB}\vert\mathcal{H}_{1}\text{ true}\right)$.
\end{definition}

\begin{proposition}
\label{prop:dist_approx_rd}
When $N_{o}\ll 1$ and $\alpha = 1-\Delta N_{o}$, the distributions of $\mathcal{V}_{L_{1}, CB}$ and $\mathcal{V}_{L_{2}, CB}$ are approximated as $\frac{1}{2^{\Theta}}\sum_{l=0}^{\Theta}{\Theta \choose l}\mathrm{G}\left(l,\frac{1}{L}\right)$ and $\frac{1}{2^{L-2\Theta}}\sum_{l=0}^{L-2\Theta}{L-2\Theta \choose l}\mathrm{G}\left(l,\frac{1}{L}\right)$, respectively. Thus, the distribution of $\frac{1}{L}\sum_{n=1}^{L-\Theta} \vert r_{D,CB,n}^{\dagger}\vert^{2}$ can be approximated as $\frac{1}{2^{L-\Theta}}\sum_{l=0}^{L-\Theta}{L-\Theta \choose l}\mathrm{G}\left(l,\frac{1}{L}\right)$.
\end{proposition}
From Proposition~\ref{prop:dist_approx_rd}, we observe that the distribution of $\vert r_{D,CB,n}^{\dagger}\vert^{2}$ is approximately same as that of the distribution of $\vert r_{D,CB,n}\vert^{2}$, for $1\leq n\leq L-\Theta$. This indicates that, Charlie solely controls the energy level on $f_{CB}$ for the first $L-\Theta$ symbols after implementing $3\phi$ DASC-MF scheme. However, during Phase-III, i.e., $L-\Theta+1\leq n\leq L$, Alice and Charlie use the energy-splitting factor $\beta^{\dagger}$ away from $1$, as shown in Table~\ref{tab:NC_beta_tab}. Therefore, when Alice transmits a burst of zeros, the energy observed at Dave on $f_{CB}$ is likely to be less than $\mathcal{E}_{CB}$, thereby resulting in higher probability of detection at Dave.  Although, Alice and Charlie can choose to use $\beta^{\dagger}$ close to $1$, this will increase the fraction of erroneous decisions of Alice's current symbols. 

While we are able to derive the probability of detection on $f_{AB}$, we do not derive closed-form expressions on probability of detection on $f_{CB}$, due to intractable density function contributed by Phase-III of the $3\phi$ DASC-MF. However, we present simulation results on detection in Sec.~\ref{ssec:simulation_covert}. 

\subsection{Covertness Analysis for ED when using SC-MAC Scheme}

Along the similar lines of Sec.~\ref{ssec:covert_ED3p}, in this section, we define the probability of miss-detection at Dave on $f_{AB}$ and $f_{CB}$, when Alice and Charlie use SC-MAC. 

\begin{definition}
\label{def:pd_MAC}
The probability of miss-detection by Dave's ED on $f_{AB}$ and $f_{CB}$ are denoted by, $\mathbf{P}_{MD,AB}^{MAC}$ and $\mathbf{P}_{MD,CB}^{MAC}$, respectively. Further, $\mathbf{P}_{MD,AB}^{MAC} = \Pr\left(\vert\mathcal{W}_{L,AB}-\mathcal{E}_{AB}\vert\leq\nu_{AB}\vert\mathcal{H}_{1}\text{ true}\right)$ and $\mathbf{P}_{MD,CB}^{MAC} = \Pr\left(\vert\mathcal{W}_{L,CB}-\mathcal{E}_{CB}\vert\leq\nu_{CB}\vert\mathcal{H}_{1}\text{ true}\right)$, such that, $\mathcal{W}_{L,AB}$ and $\mathcal{W}_{L,CB}$ are the RVs denoting the average energy of $L$ symbols received at Dave on $f_{AB}$ and $f_{CB}$, respectively.
\end{definition}
The distributions of $\mathcal{W}_{L,AB}$ and $\mathcal{W}_{L,CB}$ can be computed along the similar lines of $\mathcal{V}_{L,AB}$ and $\mathcal{V}_{L,CB}$, respectively. It is worthwhile to note that, since the entire frame of $L$ symbols is uncoordinated in energy, the sum $\mathbf{P}_{FA,AB}+\mathbf{P}_{MD,AB}^{MAC}$ and $\mathbf{P}_{FA,CB}+\mathbf{P}_{MD,CB}^{MAC}$ should be away from $1$, indicating that Dave is more likely to detect SC-MAC as compared to $3\phi$ DASC-MF scheme.

\subsection{Simulation Results}
\label{ssec:simulation_covert}
\begin{figure}
\vspace{-0.5cm}
\captionsetup{width=0.48\textwidth}
\begin{center}
    \begin{minipage}[t]{0.48\textwidth}
        \centering
        \includegraphics[scale = 0.5]{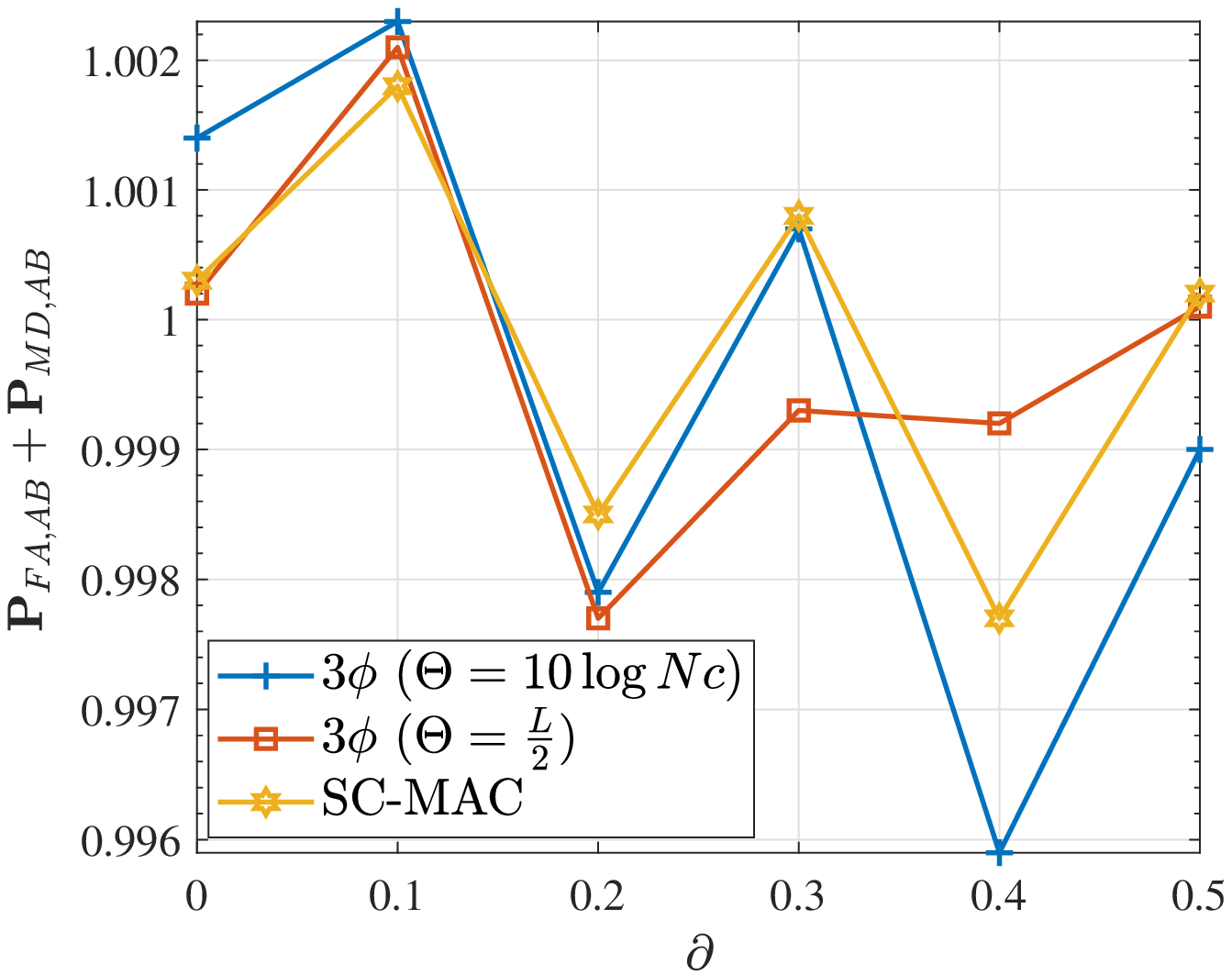}
\caption{\label{fig:sum_fab} $\mathbf{P}_{FA,AB}+\mathbf{P}_{MD,AB}$ on $f_{AB}$ as a function of $\partial$.}
    \end{minipage}%
    \hfill
   \begin{minipage}[t]{0.48\textwidth}
        \centering
        \includegraphics[scale = 0.6]{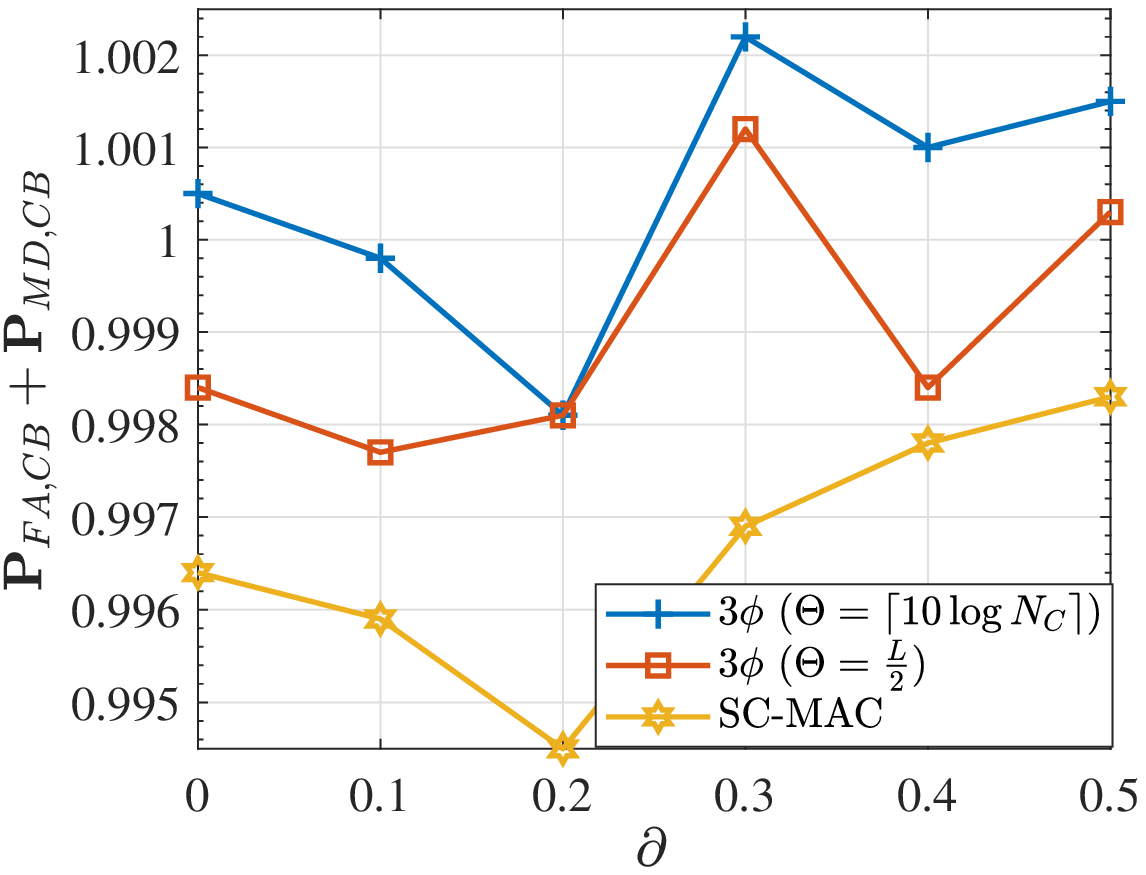}
\caption{\label{fig:sum_pmd}  $\mathbf{P}_{FA,CB}+\mathbf{P}_{MD,CB}$ on $f_{CB}$ as a function of $\partial$.}
    \end{minipage}%
    \end{center}
    \vspace{-0.5cm}
\end{figure}

In this section, we provide results to showcase the covertness of the proposed schemes. In addition to the parameters used above, we use $\nu_{AB} = \nu_{CB} = 10^{-3}$. In Fig.~\ref{fig:sum_fab}, we first plot $\mathbf{P}_{FA,AB} + \mathbf{P}_{MD,AB}$ for both the proposed schemes as a function of $\partial$ at SNR = 25 dB, and observe that the sum is close to $1$ for the considered range of $\partial$. We also observe a similar behaviour in Fig.~\ref{fig:sum_pmd}, where we plot $\mathbf{P}_{FA,CB} + \mathbf{P}_{MD,CB}$ for both the schemes. We note that, when $\Theta = \frac{L}{2}$, exactly $\frac{L}{2}$ symbols are transmitted in uncoordinated fashion, thus, $\mathbf{P}_{FA,CB} + \mathbf{P}_{MD,CB}$ is slightly higher than SC-MAC, where all the $L$ symbols are uncoordinated in energy. However, since the $3\phi$ DASC-MF scheme has minimum number of uncoordinated symbols, its $\mathbf{P}_{FA,CB} + \mathbf{P}_{MD,CB}$ is maximum and is close to $1$ for all values of $\partial\in[0,0.5]$. \looseness = -1

\section{Conclusion and Future Work}
This work presented a cooperative framework of delay-aware semi-coherent multiplex-and-forward scheme to mitigate an FD reactive jamming adversary. As a salient feature of this framework, the helper uses a practical FD radio to forward the victim's low-latency symbols to the destination. Here, the helper multiplexes the victim's symbols to its symbols to facilitate joint decoding at the destination. We first modelled the processing delay at the helper using a parameter, $\Theta$ and then showed that the symbols received from the two users arrive during different \blue{symbol intervals} at the destination, resulting in complex decoding. Further, we also pointed out that the symbols from both the users are uncoordinated in energy, leading to detection by an energy detector at the adversary. To circumvent these problems, we proposed two mitigation schemes based on the delay parameter $\Theta$. When $\Theta\leq \frac{L}{2}$, we proposed $3\phi$ delay-aware semi-coherent multiplex-and-forward scheme, wherein the legitimate users transmit their symbols using two energy-splitting factors, $\alpha$ and $\beta$, and the destination decodes them in three-phases, parametrised by $\Theta$. We also proposed a semi-coherent multiple access channel scheme, when $\Theta>\frac{L}{2}$, wherein due to large $\Theta$, the helper does not decode the victim's symbols, and the victim and the helper transmit their symbols to the destination using an energy-splitting factor, $\varepsilon$. For both the schemes, we provided  near-optimal solution on the energy-splitting factor to the optimisation problem of minimising the error rates. Finally, we showed that the victim reliably communicates with the destination while adhering to the latency constraints without getting detected by the adversary.

\blue{There are multiple directions for future work. As part of the mitigation strategy against the \emph{jam and measure} adversaries, non-coherent constellations can be designed in fast-fading channels when using delay-aware FD helper nodes. Also, throughout this work, we have assumed that the adversary can perfectly cancel its jamming energy on the victim's frequency bands. However, studying the impact of mitigation strategies with practical FD radios at the adversary is still an open problem.}

\bibliography{IEEEabrv, Ref}
%\bibliography{Ref} 
\bibliographystyle{IEEEtran}

\end{document}